\documentclass[11pt,final]{article}
\usepackage{fullpage}

\usepackage{amsmath,amssymb,amsthm}
\usepackage{booktabs} 
\usepackage{xspace}
\usepackage{accents} 
\usepackage{xparse} 
\usepackage{bbm}
\usepackage{enumitem}
\usepackage{algorithm}
\usepackage[noend]{algorithmic}

\usepackage[colorlinks=true]{hyperref}
\hypersetup{
	linkcolor=[rgb]{0.4,0,0.6},
	citecolor=[rgb]{0, 0.4, 0},
	urlcolor=[rgb]{0.6, 0, 0}
}

\usepackage{cleveref} 
\usepackage{threeparttable}
\usepackage{array}
\usepackage{booktabs}
\usepackage{multicol}
\usepackage{tabularx}
\usepackage{framed}
\usepackage{thmtools, thm-restate}
 \usepackage{mathtools}
 \usepackage[square,longnamesfirst]{natbib}
\usepackage{verbatim}

\usepackage{color}
\newcommand{\aenote}[1]{{\color{red}{#1}}}

\newcommand{\kgnote}[1]{{\color{orange}{#1}}}
\newcommand{\old}[1]{}

\newtheorem{lemma}{Lemma}
\newtheorem{proposition}{Proposition}
\newtheorem{corollary}{Corollary}
\newtheorem{theorem}{Theorem}

\theoremstyle{definition}
\newtheorem{definition}{Definition}
\newtheorem{example}{Example}

\newenvironment{numberedtheorem}[1]{%
\renewcommand{\thetheorem}{#1}%
\begin{theorem}}{\end{theorem}\addtocounter{theorem}{-1}}

\usepackage{todonotes}

\usepackage{accsupp}
\providecommand*{\napprox}{%
	\BeginAccSupp{method=hex,unicode,ActualText=2249}%
	\not\approx
	\EndAccSupp{}%
}

\newcommand{\MyFrame}[1]{\noindent \framebox[\textwidth]{ \begin{minipage}{0.97\textwidth} #1 \end{minipage}}}%

\newcommand{\fix}[1]{\textcolor{red} {\textbf{#1}}}

\newcommand{\R}{\mathbb{R}}

\newcommand{\argmax}{\operatorname{arg\,max}}

\newcommand{\prob}[2][]{\text{\bf Pr}\ifthenelse{\not\equal{}{#1}}{_{#1}}{}\!\left[#2\right]}
\newcommand{\expect}[2][]{\text{\bf E}\ifthenelse{\not\equal{}{#1}}{_{#1}}{}\!\left[#2\right]}

\newcommand{\Rev}{{\normalfont\textsc{Rev}}}

\newcommand{\opt}{{\normalfont\textsc{Opt}}}

\newcommand{\E}{\mathbb{E}}

\newcommand{\vareps}{\varepsilon}

\def\Pr{\ensuremath{\mathrm{Pr}}}

\def\argmax{\ensuremath{\mathrm{argmax}}}

\def\capx{\ensuremath{\overset{c}{\approx}}}
\def\twocapx{\ensuremath{\overset{2c}{\approx}}}

\def\kcapx{\ensuremath{\overset{kc}{\approx}}}
\def\alphaapx{\ensuremath{\overset{\alpha}{\approx}}}
\def\notcapx{\ensuremath{\overset{c}{\napprox}}}
\def\nottwocapx{\ensuremath{\overset{2c}{\napprox}}}

\crefformat{equation}{(#2#1#3)} 

\newcommand{\vecs}{\vec{s}}
\newcommand{\hypercol}{\textsf{Hypergrid-Coloring}}
\newcommand{\randhypercol}{\textsf{Random-Hypergrid-Coloring}}
\newcommand{\concavity}{\textsf{concavity}}

\newenvironment{numberedproposition}[1]{%
\renewcommand{\thetheorem}{#1}%
\begin{proposition}}{\end{proposition}\addtocounter{theorem}{-1}}

\begin{document}
\title{Interdependent Values without Single-Crossing
\thanks{This work was partially supported by the European Research Council under the European Union's Seventh Framework Programme (FP7/2007-2013) / ERC grant agreement number 337122, and by the Israel Science Foundation (grant number 317/17).}
}
\author{Alon Eden%
\thanks{%
    {Tel Aviv University (\url{alonarden@gmail.com})}}
\and Michal Feldman%
\thanks{%
    {Tel Aviv University (\url{michal.feldman@cs.tau.ac.il})}}
\and Amos Fiat%
\thanks{%
    {Tel Aviv University (\url{fiat@tau.ac.il})}}
\and Kira Goldner%
\thanks{%
    {University of Washington (\url{kgoldner@cs.washington.edu})}}
}

\maketitle

\begin{abstract}
We consider a setting where an auctioneer sells a single item to $n$ potential agents with {\em interdependent values}.
That is, each agent has her own private signal, and the valuation of each agent is a known function of all $n$ private signals.
This captures settings such as valuations for artwork, oil drilling rights, broadcast rights, and many more.

In the interdependent value setting, all previous work has assumed a so-called {\sl single-crossing condition}. Single-crossing means that the impact of agent $i$'s private signal, $s_i$, on her own valuation is greater than the impact of $s_i$ on the valuation of any other agent. It is known that without the single-crossing condition an efficient outcome cannot be obtained. 
We study welfare maximization for interdependent valuations through the lens of approximation.

We show that, in general, without the single-crossing condition, one cannot hope to approximate the optimal social welfare any better than the approximation given by assigning the item to a random bidder. Consequently, we introduce a relaxed version of single-crossing, {\sl $c$-single-crossing}, parameterized by $c\geq 1$, which means that the impact of $s_i$ on the valuation of agent $i$ is at least $1/c$ times the impact of $s_i$ on the valuation of any other agent ($c=1$ is single-crossing). Using this parameterized notion, we obtain a host of positive results.

We propose a prior-free deterministic mechanism that gives an $(n-1)c$-approximation guarantee to welfare.
We then show that a random version of the proposed mechanism gives a prior-free universally truthful $2c$-approximation to the optimal welfare for any concave $c$-single crossing setting (and a $2\sqrt{n}c^{3/2}$-approximation in the absence of concavity). We extend this mechanism to a universally truthful mechanism that gives $O(c^2)$-approximation to the optimal revenue.

\end{abstract}



\section{Introduction}

The most fundamental problem in the theory of auctions is how to sell a single item to $n$ potential buyers efficiently; i.e., how to allocate the item to the agent who values it the most. This problem has been fully resolved in 1961 in the case where agents have independent private values (IPV). Indeed, 
 \citep{vickrey1961counterspeculation} introduced the second-price auction, which is dominant-strategy incentive compatible and fully efficient.

While much research has concentrated on the independent private values model, for many high-stake auctions that arise in practice (e.g., auctions for mineral rights \citep{wilson1969communications}), agents' valuations are correlated, a special case of which are common values, where agent values are identical \citep{milgrom79,wilson1969communications}. In such cases the independent private values model is unsuitable.

To address scenarios of this type, \citet{milgrom1982theory} introduced the model of {\em interdependent} values. In their model, every agent $i$ has a non-negative private signal $s_i$, and the valuation of agent $i$ is a function of the entire signal vector; i.e., $v_i = v_i(\vecs)$.
While the signal is private and unknown by other agents or the auctioneer, the valuation functions are commonly known by all\footnote{A variation on this assumption is the {\sl asymmetric knowledge} scenario, where the valuation functions are common knowledge amongst the agents but {\sl not known} to the auctioneer. This variant was studied in \citep{dasgupta2000efficient, perry1999ex}.}.

Consider the following scenario, which is a variant of the scenario given in \citep{dasgupta2000efficient}.
\begin{example}
Two firms compete for the right to drill for oil on a given tract of land. Firm $1$ has a marginal drilling cost of $1$, while firm $2$ has a marginal drilling cost of $2$.
Suppose oil is sold in the market at a price of $4$.
Firm $1$ performs a private test and discovers that the expected size of the oil reserve is $s_1$ units.
This scenario gives rise to the following valuation functions: $v_1(s_1,s_2)=(4-1)s_1=3s_1$, and $v_2(s_1,s_2)=(4-2)s_1=2s_1$.
\label{ex:oil-drilling-sc}
\end{example}

Other examples include: (a) common value auctions, where the value of the item for sale is identical amongst bidders, but bidders may have different information about the item's value \citep{wilson1969communications}, and (b) auctions with resale, where $v_i(\vec{s}) = \sum_{j=1}^n \beta_{ij} s_j$, where $\beta_{ij} \geq 0$ for all $1\leq i,j \leq n$. Here, an agent's value depends on her personal valuation for the item and its resale value, which is reflected by the valuations of others \citep{klemperer1998auctions, myerson1981optimal}.

In scenarios with interdependent values, a direct revelation mechanism is one where every agent $i$ reports his signal $s_i$ to the mechanism, and the mechanism, knowing the (expected) valuation of each agent as a function of the signal vector, determines the allocation and the payments.
The strongest notion of truthfulness relevant in the interdependent value setting is that truth telling is an ex-post Nash equilibrium.
Truth-telling is said to be ex-post Nash if, for every bidder $i$, for every possible realization of the other bidders' signals $\vecs_{-i}$, and given that other bidders report their signals truthfully, then it is in bidder $i$'s best interest to report her true signal.
Truthfulness in dominant strategies is not viable in interdependent values since the value of one agent depends on the signals of others\footnote{To see why a dominant strategy truthful mechanism is hopeless, note that if some agent (Alice) misreports her signal, even if every other agent reports truthfully, the valuation of other agents depends on the signal given by Alice, so valuations computed using the corrupt signal from Alice (corrupt valuations) may vary from the  true valuations (based on the non-corrupted signals). Thus, the allocation and prices can be quite different than those computed from the non-corrupted signals, and it may not be a best response of the agents to reveal their true signals. In fact, it is easy to give examples where truth-telling is clearly not a reasonable response.}. 

Work on interdependent values has typically assumed a prior joint distribution (either correlated or independent) on the signals.
One can further categorize this body of work under various informational assumptions (is the distribution of signals commonly known by agents? is it known by the auctioneer?), and solution concepts (Bayes-Nash, ex-post Nash, etc.).

To quote \citet{Bergemann05}, \emph{``The mechanism design literature assumes too much common knowledge of the environment among the players and planner."}  Such common knowledge may be unrealistic, and moreover, can lead to miraculous outcomes. For example, \citet{CremerMcLean85,CremerMcLean88} show how to extract full surplus as revenue under appropriate assumptions and solution concepts (in particular, commonly known correlated distributions and Bayes Nash equilibrium).
This approach follows the seminal work of \citet{wilson1985game} that posits {\em Wilson's doctrine}\footnote{Due to Robert Wilson (1985), and not referring to the Harold Wilson wiretapping doctrine (1966), nor to the Woodrow Wilson extension of the Monroe doctrine (1912).}: simple detail-free mechanisms should be preferred in order to alleviate the risks introduced by various assumptions.

In this work we follow this doctrine and propose prior-free mechanisms---mechanisms that do not require any knowledge of the underlying distribution. Moreover, our mechanisms are universally ex-post truthful. Ex-post truthfulness means that bidders maximize utility for any signal profile, so any underlying distribution of bidder values is irrelevant. For universally truthful mechanisms, even if the bidders know in advance the random internal coin tosses made by the mechanism they still have no incentive to bid non-truthfully.


\vspace{0.1in}
\noindent {\bf Single-crossing.} All previous work on interdependent values
assumes some version of a {\sl single-crossing (SC) condition} on the valuations \citep{milgrom1982theory,Aspremont82,maskin1992,ausubel1999generalized,dasgupta2000efficient,
bergemann2009information,CFK,che2015efficient,li2016approximation,RTCoptimalrev}. Several definitions of single-crossing have been suggested in the literature. We define the single-crossing (SC) condition as done in \cite{RTCoptimalrev}\footnote{Appendix  \ref{sec:alternative-definitions} briefly discusses other definitions of single-crossing, and how the results of our paper can be adapted appropriately.}: for every agent $i$, for any set of other players'  signals $\vec{s}_{-i}$, and for every agent $j$,
$$
\frac{\partial v_i(s_i, \vec{s}_{-i})}{\partial s_i} \geq \frac{\partial v_j(s_i, \vec{s}_{-i})}{\partial s_i}.
$$
Thus, as a function of signal $i$ ($s_i$), the valuation of bidder $i$ ($v_i$), increases at least as much as the valuation of any other bidder ($v_j$, $j\neq i$).

Under the SC condition, a generalization of the VCG mechanism maximizes social welfare \citep{maskin1992,ausubel1999generalized,RTCoptimalrev,CFK}; this generalized VCG mechanism is deterministic, ex-post truthful and prior-free.
It can be easily verified that the valuations in the oil drilling example described in Example~\ref{ex:oil-drilling-sc} satisfy the SC condition. Indeed, $v_1$ (as a function of $s_1$) grows faster than $v_2$ (as a function of $s_1$).

However, there are many scenarios where the SC condition does not hold.
Consider for example the following scenario of two firms competing for oil drilling rights, given in \citep{dasgupta2000efficient}.

\begin{example}
\label{ex:oil-drilling-no-sc}
As before, each firm has a marginal cost for drilling, but this time each firm also has a fixed cost for drilling.
Firm $1$ has a fixed cost of $1$ and a marginal cost of $2$, while firm $2$ has a fixed cost of $2$ and a marginal cost of $1$.
Suppose oil is sold in the market at a price of $4$.
Firm $1$ performs a private test and discovers that the expected size of the oil reserve is $s_1$ units.
This scenario gives rise the following valuation functions: $v_1(s_1,s_2)=(4-2)s_1 - 1 = 2s_1-1$, and $v_2(s_1,s_2)=(4-1)s_1 - 2=3s_1-2$.
These valuations do not satisfy the SC condition.
The increase in firm $2$'s valuation increases faster than firm $1$'s valuation as a function of $s_1$.
\end{example}

As another example, consider the following scenario regarding two retail chains, each interested in renting some location for a shop.

\begin{example}
Each of the two retail chains has conducted a survey to estimate the income level of the population in the area. Let $s_i$ be the estimate obtained by retail chain $i$ (furthermore, assume that $1\leq s_i\leq 2$).
A good estimate for the income level in the area is the average of $s_1$ and $s_2$.
The valuations of the two retail chains for opening a shop at this location are functions of the income level in the area, since the demand for the goods is a function of the customer income level.
Suppose that the first retail chain sells normal goods, whereas the second retail chain sells luxury goods. Let the corresponding valuations be  $v_1(s_1,s_2)=0.06+\frac{s_1+s_2}{2}$ (for normal goods) and $v_2(s_1,s_2)=(\frac{s_1+s_2}{2})^{1.1}$ (for luxury goods).
These valuations do not satisfy the SC condition: as a function of $s_1$, the rate of increase for $v_2$ is faster than the rate of increase for $v_1$. (See Figure \ref{fig:normal-vs-luxury}).
\label{ex:normal-vs-luxary}
\end{example}

\begin{figure}[H]
	\centering
	\includegraphics[scale=0.3]{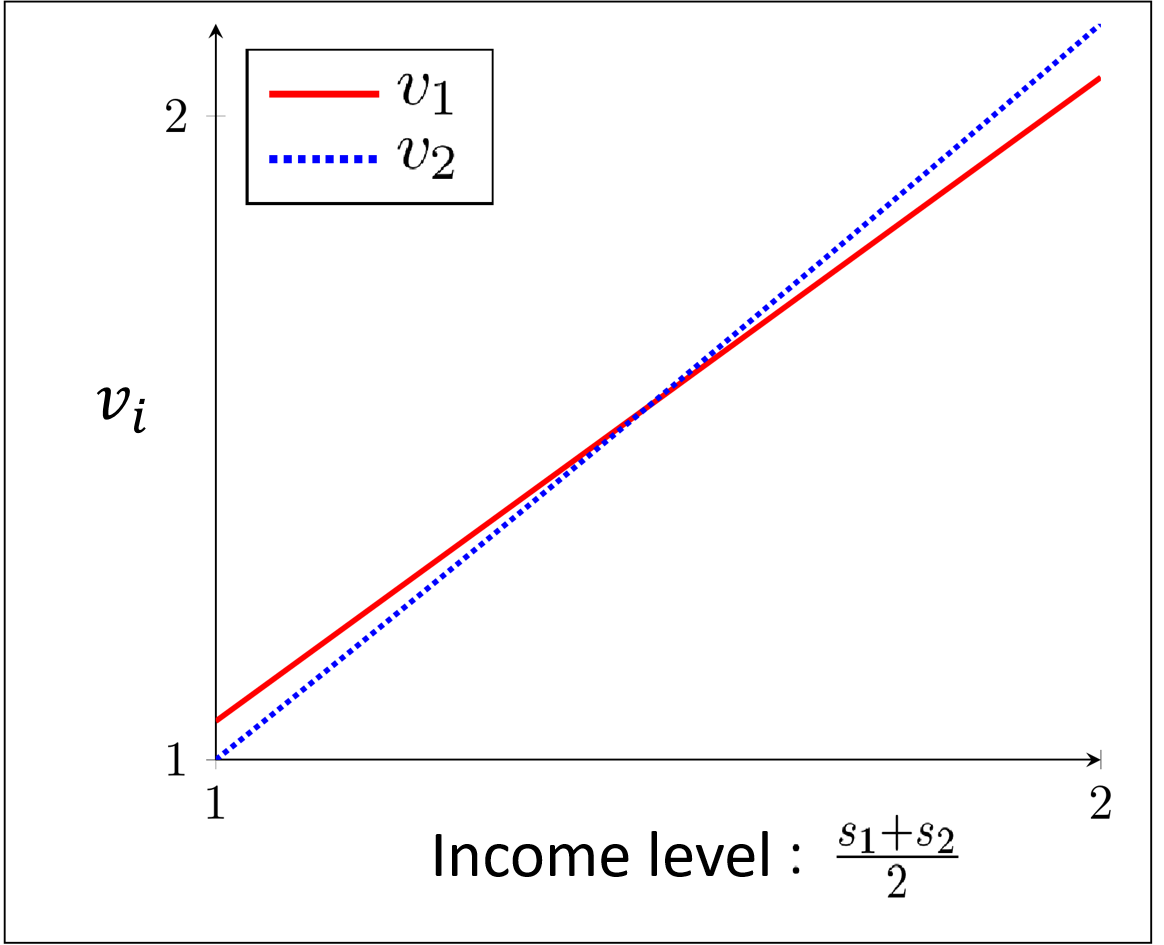}
	\caption{Valuations $v_1$ (retail chain 1, normal goods) and $v_2$ (retail chain 2, luxury goods) for opening a shop as a function of the income level. Income level  is itself well-estimated as the average of the two signals, $s_1$ and $s_2$.}
	\label{fig:normal-vs-luxury}			
\end{figure}




It is well known that, without the single-crossing assumption, it is generally impossible to achieve full efficiency  \citep{dasgupta2000efficient, jehiel2001efficient}. 
Thus, the next step is to seek near-optimal guarantees. This is our approach; namely, we study settings with interdependent valuations through the lens of approximately-optimal efficiency.

Unfortunately, without single-crossing, it is not generally possible to give {\em any} non-trivial approximation guarantee to the social welfare (see Section \ref{sec:detimpossibility}).
Thus, the obvious next question is to study the tradeoff between assumptions and the quality of the approximation.
The two ends of this spectrum are (a) the single-crossing assumption and full efficiency, and (b) no assumptions on the valuation functions and no efficiency guarantees.

To study this tradeoff, we introduce the following parameterized version of the single-crossing condition:
A valuation profile is said to be \emph{$c$-single-crossing} if for every agent $i$, for any set of other players'  signals $\vec{s}_{-i}$, and for every agent $j$,
$$
c \cdot \frac{\partial v_i(s_i, \vec{s}_{-i})}{\partial s_i} \geq \frac{\partial v_j(s_i, \vec{s}_{-i})}{\partial s_i}.
$$
This means that as buyer $i$'s signal changes, the change in any other agent's value is at most $c$ times the change in agent $i$'s value.
For example, the valuations described in Example \ref{ex:oil-drilling-no-sc} violate the SC condition, but they are $1.5$-SC.
Similarly, the valuations described in Example \ref{ex:normal-vs-luxary} are $1.18$-SC.
Can these scenarios, which deviate only slightly from SC, permit good approximation guarantees?
We answer this question in the affirmative.
In particular, we provide welfare guarantees that depend on the parameter $c$.
Our approximation guarantees are strong in the sense that they hold for any signal profile $\vec{s}$. (This is in contrast to approximation guarantees that are given in expectation over the signal profile).

%


Moreover, we are also interested in revenue guarantees under relaxed SC valuations.
In \citep{CFK} it was shown that for SC and concave valuations, revenue maximization follows from welfare
maximization. Here, concavity essentially means that buyers are more sensitive
to a change in some agent's signal when all of the buyers have lower signals.
Like \citet{CFK}, we show that, for settings with $c$-single-crossing and concave valuations, [approximate] revenue
maximization follows from [approximate] welfare maximization.



\subsection{Our Results}

Before presenting our results, we present our approximation notion.

\paragraph{Our approximation notion}

We say that the value of bidder $j$ is \emph{$\alpha$-approximated} by bidder $i$'s value at profile $\vec{s}$, which we denote by $v_i(\vec{s}) \alphaapx v_j(\vec{s})$, if $v_j(\vec{s}) \leq \alpha v_i(\vec{s})$.
We say that an allocation gives a \emph{prior-free} \emph{$\alpha$-approximation} to welfare if, for every signal profile $\vec{s}$, the value of the bidder with the highest value is $\alpha$-approximated by the bidder being allocated to in $\vec{s}$ by the mechanism.
A prior-free mechanism gives an $\alpha$-approximation to the welfare without any assumption on prior distributions.
The quality of the approximation, for our randomized mechanisms, is in expectation over the internal
coin tosses of the mechanism and nothing else\footnote{The alternative to prior-free approximation is that the approximation only holds in expectation over some prior distribution; i.e.,
$\alpha \sum _{\vec{s} \in \times_i S_i } f(\vec{s}) \sum_i x_i(\vec{s}) v_i(\vec{s}) \geq \sum _{\vec{s} \in \times_i S_i } f(\vec{s}) \max _i v_i(\vec{s})$, where $\vec{s}$ comes from a joint distribution $F$ with density $f(\cdot)$.}.


\begin{table}[h!]
\begin{threeparttable}[h!]
{\centering \footnotesize
\begin{tabular}{|l r|c r|c r|c r| c r|} \hline
 \multicolumn{2}{|c|}{Setting}
& \multicolumn{2}{c|}{\begin{tabular}{l} Deterministic\\Upper Bound \\Prior-free\end{tabular}}
& \multicolumn{2}{c|}{\begin{tabular}{l} Randomized\\Upper Bound\\Prior-free \end{tabular}}
& \multicolumn{2}{c|}{\begin{tabular}{l} Deterministic\\Lower Bound \\Prior-free\end{tabular}}
& \multicolumn{2}{c|}{\begin{tabular}{l} Randomized\\Lower Bound\\(even w/ Prior) \end{tabular}} \\
\hline
\multicolumn{2}{|c|}{\begin{tabular}{l} Without Single\\Crossing \end{tabular}} & \multicolumn{2}{c|}{---} &  $n$ & $^{(\S 1)}$&
$\infty$ &$^{(\ddag 1)}$  & $n$ &$^{(\sharp 1)}$  \\ \hline
\multicolumn{2}{|c|}{\begin{tabular}{l} $c$-single-crossing\\arbitrary $n$, $k$ \end{tabular}} &
$(n-1)c$ & $^{(\dag 2)}$ &
$O\left( \min(c^{3/2} \cdot \sqrt{n},n)\right)$ & $^{(\S 2)}$ &
$c+\Omega(1)$ &$^{(\ddag 2)}$ &
c & $^{(\sharp 2)}$   \\ \hline
\multicolumn{2}{|c|}{\begin{tabular}{l} $c$-single-crossing \\ arbitrary $n$,  $k$\\Concavity \end{tabular}} &
 $(n-1)c$ &$^{(\dag 3)}$  &
  $\min(2c,n)$&   $^{(\S 3)}$& $c$
  &$^{(\ddag 3)}$ &  &
  \\ \hline
\multicolumn{2}{|c|}{\begin{tabular}{l} $c$-single-crossing \\ 2 bidders ($n=2$)\end{tabular}} & $c$ &$^{(\dag 4)}$ &
 ${\min(c,2)}$ &$^{(\S 4)}$ &
  $c$ & $^{(\ddag 4)}$&
     &
    \\ \hline
\multicolumn{2}{|c|}{\begin{tabular}{l} $c$-single-crossing \\ 2 Signals ($k=2$)\end{tabular}} & $c$ &$^{(\dag 5)}$ &
 $\min(c,n)$ & $^{(\S 5)}$ &
  $c$ &$^{(\ddag 5)}$&
  $c$ & $^{(\sharp 5)}$\\ \hline
 \end{tabular}
 \begin{tablenotes} \footnotesize
 \setlength{\columnsep}{0.8cm}
\setlength{\multicolsep}{0cm}
\begin{framed}
  \begin{multicols}{2}
     \item[$(\S 1 -- \S 5)$] Choosing a bidder at random gives an $n$ approximation.
     \item[$(\ddag 1)$] Section \ref{sec:detimpossibility}.
     \item[$(\sharp 1)$] Section \ref{sec:impossibility}.
     \item[$(\dag 2)$] Via the Hypergraph-Coloring algorithm  of Section \ref{sec:nbidders}.
     \item[$(\S 2)$] The $c^{3/2} \cdot \sqrt{n}$ approximation --- Section \ref{sec:random_sqrt}.
     \item[$(\ddag 2)$] Subsection \ref{sec:nocapprox}.
     \item[$(\sharp 2)$] For sufficiently large $n$, implied by cell $(\sharp 5)$, lower bound example with 2 signals.
     \item[$(\dag 3)$] Implied by cell $(\dag 2)$.
     \item[$(\S 3)$] The $2c$ approximation --- Section \ref{sec:random}. If we only have $d$-concavity, for some $d\geq 1$, then $2c$ becomes $c(d+1)$.
     \item[$(\ddag 3)$] Lower bound in Figure \ref{fig:c-is-tight} has concave valuations (with 2 bidders and 2 signals).
     \item[$(\dag 4)$] See Section \ref{subsec:2bidders}.
     \item[$(\S 4)$] A deterministic $c$ approximation is implied by cell $(\dag 4)$.
     \item[$(\ddag 4)$] Lower bound in Figure \ref{fig:c-is-tight} (with 2 bidders and 2 signals).
     \item[$(\dag 5)$] Using the algorithm of Figure \ref{fig:twobiddercoloring}, Section \ref{subsec:2bidders}.
     \item[$(\S 5)$] A deterministic $c$ approximation is implied by cell $(\dag 5)$.
     \item[$(\ddag 5)$] Implied by cell $(\sharp 5)$, for sufficiently large $n$.
     \item[$(\sharp 5)$] For sufficiently large $n$, see Subsection \ref{sec:random_c_lb}.
     \end{multicols}
  \end{framed}
   \end{tablenotes}}

\end{threeparttable}
\caption{Our Results for Social Welfare maximization. The randomized upper bounds are also universally truthful.}
\label{tab:results}
\end{table}

\paragraph{Our results}

In Section~\ref{sec:detimpossibility}, we show that without SC, no deterministic prior-free mechanism can achieve {\em any} guarantee for welfare.
Moreover, as we show in Section~\ref{sec:impossibility}, there exists a distribution on signals such that no (randomized) mechanism can do better than allocating the item to a random agent.
Such a mechanism would result in no better than $1/n$ fraction of the optimal welfare.

We then study welfare approximation in settings where buyer valuations satisfy $c$-SC.
We also consider the impact of the number of possible signals per buyer, {\sl i.e.}, the size of the set of potential signals that a buyer may have.
Some of our results depend on this parameter.

In Section~\ref{sec:warmups}, we identify two settings that admit a deterministic prior-free $c$-approximation: (a) $2$ bidders, any signal space (Theorem~\ref{thm:2bidder}), and (b) any number of bidders, each with a signal space of size $2$ (Theorem~\ref{thm:2signals}).  This approximation is tight in several senses.  First, it is provably impossible to obtain better than a $c$-approximation in these settings.  In fact, in the case of signal spaces of size $2$, this impossibility holds even if one considers randomized, truthful-in-expectation mechanisms with a known prior distribution on the buyer signals.
We also show that these are the most general settings that admit a $c$-approximation.  In particular, for $3$ bidders, one with signal space of size $3$ and two with signal spaces of size $2$---no $c$-approximation on the social welfare is possible (Proposition \ref{prop:superc}).


We also give results for general settings with $n$ bidders, each of which has one one of $k$ possible signals.
We start by constructing a family of deterministic prior-free mechanisms that obtain a $c(n-1)$-approximation (Theorem~\ref{thm:hypercolapx}).
The mechanism imposes some (arbitrary) ordering $\pi$ over the bidders.
In every iteration, the next bidder in $\pi$ is added, and a ``tentative allocation" is determined. This includes making
appropriate changes to the previous tentative allocation so as to preserve approximation guarantees  and monotonicity (truthfulness)  with respect to the newly added bidder.

The algorithm is described as computing the full allocation table for all possible signal profiles, which would take non-polynomial time. However, we show that computing the allocation for any profile of signals, on the fly, can be done in polytime ($O(n^2k\log k)$ time---Theorem~\ref{thm:polytime}).

We then show that a randomized version of this mechanism has much better welfare guarantees.
In particular, if the valuations are concave, then the randomized mechanism obtained by imposing a random ordering over the agents gives a prior-free, universally truthful mechanism that gives a $2c$-approximation (Theorem~\ref{thm:concave-2c}).
(This guarantee extends to $c(d+1)$ for $d$-concave valuations---a parameterized version of concavity).
For general $c$-SC valuations, this mechanism gives a $2c^{3/2}\sqrt{n}$-approximation (Theorem~\ref{thm:randrootn}).

While the main focus of this paper is social welfare guarantees, our results have implications to revenue optimization.
For concave valuations, we establish a black-box reduction from (approximate) welfare to (approximate) revenue, for every implementable allocation rule. To establish this reduction, we use similar ideas to ones used in \citep{CFK}.  In particular, every $c$-single-crossing setting with $n$ bidders and concave valuations admits a randomized, universally truthful mechanism that gives an $O(c^2)$-approximation (and $O(c^2d^2)$-approximation for $d$-concave valuations) (Theorem~\ref{thm:revapx}).

We refer the reader to Table~\ref{tab:results}, where we give an in depth description of the implications of our results for social welfare. The table shows how the various results relate to one another and gives references to where each upper and lower bound stems from in the organization of this paper.

\section{Model and Preliminaries} 
\label{sec:preliminaries}

We consider an auction setting where a single item is sold to $n$ agents with \emph{interdependent} values (\citet{milgrom1982theory}).
Each agent $i \in \{1, \ldots, n\}$ receives a single signal $s_i$ which is known only to agent $i$. Let $\vec{s}=(s_1,s_2,\ldots,s_n)$ be a signal profile, let $\vec{s}_{-i}$ denote all signals but $s_i$, and let $(s_i',\vec{s}_{-i})$ denote the profile $\vec{s}$ but where $s_i$ has been replaced with $s'_i$. The set of potential signals for bidder $1 \leq i\leq n$ is a discrete signal space $S_i$. Without loss of generality, assume $S_i=\{0,1,\ldots,k_i\}$; for ease of exposition we may assume that $k_i=k$ for all $i$.

Each agent $i$ also has a publicly known valuation function $v_i: \times_i S_i \rightarrow \mathbb{R}_{\geq 0}$, which maps every signal profile of the $n$ agents to a real (non-negative) value.  The valuation functions for all bidders $i$ are monotone non-decreasing in every signal $s_j$ for all $j$.


The input to a mechanism is a vector of reported signals $\vec{b}=( b_1, b_2, \ldots, b_n)$.  Mechanisms are described by  a pair $(x,p)$, where $x$ is a set of allocation functions $x=\{x_1(\vec{b}),\ldots,x_n(\vec{b})\}$ satisfying $\sum_i x_i(\vec{s}) \leq 1$ for all possible $\vec{b}$, and $p$ a set of payment functions  $p=\{p_1(\vec{b}),\ldots,p_n(\vec{b})\}$.
An allocation function $x_i:\times_j S_j\rightarrow [0,1]$ maps every bid profile $\vec{b}$ to the probability that agent $i$ gets allocated.  A payment rule $p_i: \times _j S_j \rightarrow \R$ maps the reported bids $\vec{b}$ to the expected payment from bidder $i$. A bidder's expected utility is quasilinear, given by $x_i (\vec{b}) \cdot v_i(\vec{s}) - p_i(\vec{b})$ where $\vec{s}$ is the true signal profile of the agents.

\subsection{Solution Concepts}

We focus on the design of incentive-compatible and individually rational mechanisms.
In the interdependent setting, we cannot hope for truth-telling to be a dominant strategy: one agent's misreport could cause the auctioneer to overcharge a different agent. The strongest incentive-compatibility (IC) concept in this setting is thus that truth-telling is an ex-post Nash Equilibrium, or that it is in every agent $i$'s best interest to report his true signal $s_i$ given that all other agents reported the true signal profile $\vec{s}_{-i}$:
$$x_i(\vec{s}) \cdot v_i(\vec{s}) - p_i(\vec{s}) \geq x_i(b_i, \vec{s}_{-i}) \cdot v_i(\vec{s}) - p_i(b_i, \vec{s}_{-i})  \quad \quad \quad  \forall \vec{s} \in \times _{j} S_j, b_i \in S_i. \quad\quad\quad [\text{IC}]$$
Similarly, individual rationality (IR) cannot possibly hold if signals are corrupted, so the appropriate notion with respect to individual rationality is that of ex-post IR, {\sl i.e.}, $$x_i(\vec{s}) \cdot v_i(\vec{s}) - p_i(\vec{s}) \geq 0  \quad \quad \quad  \forall \vec{s} \in \times _{j} S_j, b_i \in S_i.\quad\quad\quad [\text{IR}]$$

Thus, in this paper, incentive-compatibility (IC) refers to truth-telling being an ex-post Nash and individual rationality (IR) refers to ex-post individually rational. We use the term {\em truthful} mechanism for mechanisms that are both; all the mechanisms we present are truthful.

We emphasize that the information state is the following: (a) agents know their own valuation function $v_i$, and their own private signal $s_i$; (b) the auctioneer knows the valuation functions of the agents participating; and (c) signals are private and arbitrary. Except for bidder $i$, no other bidder, nor the auctioneer, knows {\sl anything} about $s_i$.

An allocation $x=(x_1,\ldots,x_n)$ is said to be {\sl implementable}
if there exist payment functions $p=p_1,\ldots,p_n$ such that the mechanism $(x,p)$ is truthful.

We give both deterministic and randomized truthful mechanisms. Our randomized mechanisms are a distribution over a family of deterministic mechanisms, all of which are truthful (i.e., they are {\em universally truthful} \citep{NisanRonen99,DobzinskiDughmi09}).

All the mechanisms we present are {\em prior-free}; i.e., there is no assumption of an underlying prior over the agents' signals; neither in the design of the mechanism, nor in the truthfulness notion, nor in the approximation guarantees. Every assertion holds for every signal profile $\vec{s}$.

We note that weaker solution concepts appear in the literature; i.e., truthful-in-expectation (vs. universally truthful), Bayesian truthful (vs. ex-post truthful), and interim IR (vs. ex-post IR). The reader is referred to \citep{RTCoptimalrev} for formal definitions. All of our positive results hold for our solution concept, which is analogous to dominant-strategy IC in the private value setting. Many of our impossibility results hold even with respect to weaker solution concepts.

\subsection{Monotone Allocations}

The following definition is key in characterizing which mechanisms are truthful.
\begin{definition}[Monotonicity]
	An allocation function $x_i$ is said to be monotone if for every $\vec{b}_{-i}$, $x_i(\vec{b}_{-i},b_i)$ is monotone non-decreasing in the signal $b_i$.
\end{definition}

Similar to Myerson's characterization for the independent private value setting, \citet{RTCoptimalrev} characterized the class of truthful mechanisms, as follows.

\begin{proposition}
\label{prop:char-ic}
Monotonicity is a necessary and sufficient condition for allocation functions $x$ to be \emph{implementable}, {\sl i.e.}, there exist payment functions $p$ such that the mechanism $(x,p)$ is truthful.  Moreover, an analogue of Myerson's payment identity holds, so the payment is uniquely determined by the allocation function.
\end{proposition}

It follows that constructing a truthful mechanism is equivalent to constructing a monotone allocation function.

An allocation function $x_i$ is called {\sl deterministic} if $x_i(\vec{b})\in \{0,1\}$ for all $i$ and all $\vec{b}$.
For a deterministic mechanism, we use the notation $$x(\vec{s}) = i \quad \quad\quad \text{when} \quad\quad\quad x_j(\vec{s}) = \begin{cases} 1 & \text{if } j=i \\ 0 & \text{otherwise.} \end{cases}$$
Given a deterministic monotone allocation function $x_i$ and signals for agents $\neq i$, $\vec{s}_{-i}$, the {\sl critical signal} for $i$ is as follows: if there exists some $b_i$ such that $x_i(s_{-i},b_i)=1$ then set $b_i^*=\min_{b_i} x_i(s_{-i},b_i)=1$, otherwise there is no critical signal for $i$.

For deterministic truthful mechanisms, the payment identity of \citet{RTCoptimalrev} implies the following.
\begin{proposition}\label{prop:deterministic_payment}
	Let agent $i$ be the allocated winner at bid profile $\vec{s}$ in a deterministic truthful mechanism. Then his payment is his value at the critical bid, i.e., $p_i=v_i(\vec{s}_{-i},b_i^*)$.
\end{proposition}


\subsection{Single-Crossing}

A single-crossing condition captures the idea that bidder $i$'s signal has a greater effect on bidder $i$'s value than on any other bidder's value. Formally:

\begin{definition}[Single-Crossing]
	A valuation profile is said to satisfy the single-crossing condition if for every agent $i$, for any set of other players'  signals $\vec{s}_{-i}$, and for every agent $j$, $$\frac{\partial v_i(s_i, \vec{s}_{-i})}{\partial s_i} \geq \frac{\partial v_j(s_i, \vec{s}_{-i})}{\partial s_i}.$$
\end{definition}

In the context of discrete signal spaces, for $s_i = 1, \ldots, k_i$, define $\frac{\partial v_j(s_i, \vec{s}_{-i})}{\partial s_i} = v_j(s_i, \vec{s}_{-i}) - v_j(s_i - 1, \vec{s}_{-i})$.


Whenever single-crossing holds, full efficiency can be achieved: once an agent has the highest value, by the single-crossing condition, his value continues to be the highest of all bidders as his signal increases. Therefore, allocating to the bidder with the highest value defines a monotone allocation rule, and therefore, according to Proposition~\ref{prop:char-ic}, it is implementable. The payment of that agent is then just his value at his critical signal. Note also that this mechanism is deterministic and prior-free.  This is precisely the generalized VCG mechanism used in \citep{maskin1992}.

Unfortunately, according to Proposition~\ref{prop:char-ic}, monotonicity of the allocation rule is also necessary. Hence, without single-crossing, it is impossible to have a truthful mechanism that maximizes welfare.  

\subsection{Approximation}

While full efficiency is unattainable without single-crossing, one might hope for approximate efficiency.
We say that the value of bidder $j$ is \emph{$\alpha$-approximated} by bidder $i$'s value at profile $\vec{s}$, which we denote by $v_i(\vec{s}) \alphaapx v_j(\vec{s})$, if $v_j(\vec{s}) \leq \alpha v_i(\vec{s})$.  We say that an allocation $x$ gives a \emph{prior-free} \emph{$\alpha$-approximation} to welfare if
$$\alpha \sum_i x_i(\vec{s}) v_i(\vec{s}) \geq \max _i v_i(\vec{s}) \quad\quad\quad  \forall \vec{s} \in \times_i S_i.$$

A prior-free mechanism gives an $\alpha$-approximation to the welfare without any assumption on prior distributions.
The quality of the approximation (welfare, revenue), for our randomized mechanisms, is in expectation over the internal
coin tosses of the mechanism and nothing else\footnote{The alternative to prior-free approximation is that the approximation only holds in expectation over some prior distribution; i.e.,
$\alpha \sum _{\vec{s} \in \times_i S_i } f(\vec{s}) \sum_i x_i(\vec{s}) v_i(\vec{s}) \geq \sum _{\vec{s} \in \times_i S_i } f(\vec{s}) \max _i v_i(\vec{s})$, where $\vec{s}$ comes from a joint distribution $F$ with density $f(\cdot)$.}.

The generalized VCG mechanism (that gives optimal efficiency under single-crossing) is a prior-free mechanism.

\subsection{Impossibility Results for Settings Without Single-Crossing}
\label{sec:detimpossibility}

Here we show that no truthful, prior-free, and deterministic mechanism can obtain \emph{any} bounded approximation ratio when the valuations do not satisfy single-crossing. \\
\\
\noindent {\bf Example:} [Impossibility for deterministic prior-free mechanisms]
Consider a scenario with two bidders (bidder $1$ and bidder $2$), where $S_1=\{0,1\}$ and $S_2=\{0\}$, and the following valuation functions:
\begin{align*}
v_1(s_1=0,s_2=0) = r ; \quad \quad \quad & v_1(s_1=1,s_2=0) = r ;\\
v_2(s_1=0,s_2=0) = 1 ; \quad \quad \quad & v_2(s_1=1,s_2=0) = r^2.
\end{align*}

It is easy to see that $v_1$ does not satisfy single-crossing since when $s_1$ increases, $v_1$ does not increase but $v_2$ increases by $r^2 -1$, making $v_1$ go from being $r$ times greater than $v_2$ to being $r$ times smaller than it.

We claim that, for these valuations, no truthful, deterministic, and prior-free mechanism has an approximation ratio better than $r$.
To see this, consider the signal profile $(s_1=0,s_2=0)$. To get a better than $r$-approximation for this profile, bidder $1$ must win the item. Truthfulness requires the allocation to be monotone in each bidder's signal, hence bidder 1 must also win at report $(s_1=1,s_2=0)$, which results in an allocation that is a factor of $r$ off from the optimal allocation. Since $r$ is arbitrary, the approximation ratio is arbitrarily bad.

In Section \ref{sec:impossibility} we show that without single-crossing, even when we consider randomized mechanisms, no truthful randomized mechanism can achieve a better approximation than the simple mechanism that allocates the item to a random bidder, disregarding the reported signals altogether. This impossibility holds even if signals come from a known joint distribution.

\subsection{Approximate Single-Crossing and Its Implications}

The impossibility results motivate the following relaxed notion of single-crossing.

\begin{definition}[$c$-Single-Crossing]
A valuation profile is said to satisfy the $c$-single-crossing condition if for every agent $i$, for any set of other players'  signals $\vec{s}_{-i}$, and for every agent $j$, $$c\cdot \frac{\partial v_i(s_i, \vec{s}_{-i})}{\partial s_i} \geq \frac{\partial v_j(s_i, \vec{s}_{-i})}{\partial s_i}.$$
\end{definition}

We now explore some useful properties of $c$-single-crossing, depicted in Figure~\ref{fig:contapx}.

\begin{lemma} \label{lem:contapx}
For any profile $\vec{s}$ and $\alpha\geq c$, if $v_i(\vec{s}) \alphaapx v_j(\vec{s})$ and the valuations satisfy $c$-single-crossing, then for any $\vec{s'} = (s_i', \vec{s}_{-i})$ such that $s_i' > s_i$,
$v_i(\vec{s'}) \alphaapx v_j(\vec{s'})$.
\end{lemma}

\begin{proof}
By assumption, $\alpha v_i(\vec{s}) \geq v_j(\vec{s})$.  By $c$-single-crossing, $c \left( v_i(\vec{s'}) - v_i(\vec{s}) \right) \geq v_j(\vec{s'}) - v_j(\vec{s})$.  In addition, $\alpha \geq c$, so $\alpha v_i(\vec{s'}) \geq \alpha v_i(\vec{s}) + c \left( v_i(\vec{s'}) - v_i(\vec{s}) \right)$.  Then adding the first two inequalities and using this fact gives the desired result: $\alpha v_i(\vec{s'}) \geq v_j(\vec{s'})$.
\end{proof}

\begin{corollary} \label{cor:contapx} For valuations that satisfy $c$-single-crossing, the following are implied by Lemma~\ref{lem:contapx}:
\begin{itemize}
\item The contrapositive: For any $\alpha \geq c$, if $v_i(\vec{s}) \not \alphaapx v_j(\vec{s})$, then at any profile $\vec{s'} = (s_i', \vec{s}_{-i})$ for $s_i' < s_i$, $v_i(\vec{s'}) \not \alphaapx v_j(\vec{s'})$.  See Figure~\ref{fig:contapx}.
\item For any $\alpha \geq c$, if $v_i(\vec{s}) \geq v_j(\vec{s})$, then at any profile $\vec{s'} = (s_i', \vec{s}_{-i})$ for $s_i' > s_i$, $v_i(\vec{s'}) \alphaapx v_j(\vec{s'})$.
\end{itemize}
\end{corollary}

\begin{figure}
\begin{minipage}[t]{0.7\textwidth}
\centering
\includegraphics[scale=.35]{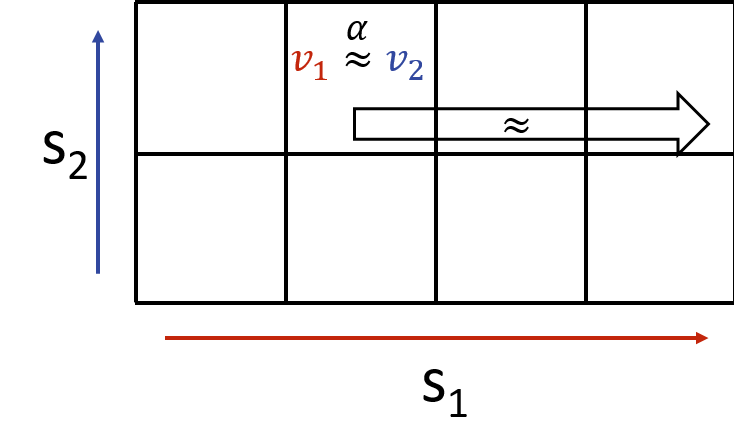} \quad\quad \includegraphics[scale=.35]{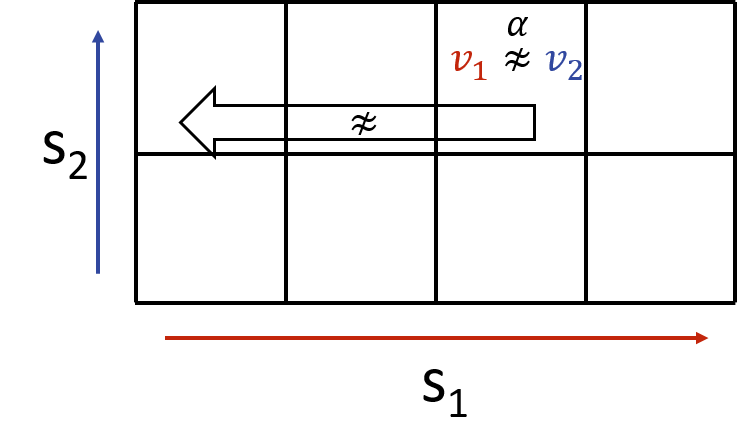}
\caption{Left: An illustration of Lemma~\ref{lem:contapx}.  Right: An illustration of the contrapositive, mentioned in Corollary~\ref{cor:contapx}.} \label{fig:contapx}
\end{minipage}
\hfill
\begin{minipage}[t]{0.25\textwidth}
\centering
\includegraphics[scale=.3]{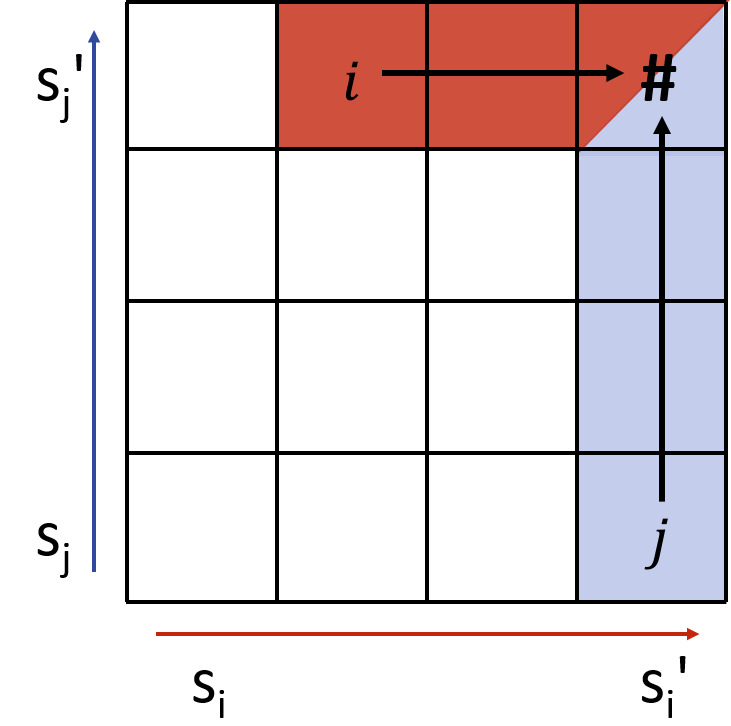}
\caption{An illustration of a propagation conflict.} \label{fig:propconflict}
\end{minipage}
\end{figure}

\paragraph{Remark.}
We have defined our notion of $c$ single-crossing as a relaxation of the single-crossing definition in \citep{RTCoptimalrev}. 
If a set of valuation functions obey the  \citep{RTCoptimalrev} definition then they also obey the various other definitions of single-crossing used in \citep{milgrom1982theory,Aspremont82,maskin1992,ausubel1999generalized,dasgupta2000efficient,
bergemann2009information,CFK,che2015efficient,li2016approximation}, but not vice versa. 
A discussion regarding alternative definitions of single-crossing and the applicability of our results under these definitions appears in Appendix \ref{sec:alternative-definitions}.

\subsection{Concave valuations}



\begin{definition}[Concave Valuations] \label{def:concavity}
Valuations are said to be concave if for all bidders $i, j$ and for any $\vec{s}_{-j}$ and $\vec{s'}_{-j}$ such that $(\vec{s}_{-j})_\ell \leq (\vec{s'}_{-j})_\ell$ for all $\ell \neq j$, it holds that
$$ \frac{\partial}{\partial s_j} v_i(\vec{s}_{-j}, s_j) \geq \frac{\partial}{\partial s_j} v_i(\vec{s'}_{-j}, s_j).$$
\end{definition}

That is, when anyone aside from $j$'s signals are lower, then every bidder $i$ 
is more sensitive to the change in $j$'s signal as when those bidders (aside from $j$) have higher signals.

\vspace{0.1in}

\noindent {\bf Remark}[$d$-concave valuations]: We also consider a parameterized version of concavity. Valuations are said to be $d$-concave if the inequality above is replaced by
$$ d \cdot \frac{\partial}{\partial s_j} v_i(\vec{s}_{-j}, s_j) \geq \frac{\partial}{\partial s_j} v_i(\vec{s'}_{-j}, s_j).$$
That is, every bidder is at least $1/d$ more sensitive to a change in $j$'s signal when bidders' signals are lower.

\subsection{Monotonicity and propagation}

We prove that several mechanisms are truthful and prior-free approximations to social welfare.  In each proof, we have two components: monotonicity and approximation.  To demonstrate that the approximation is prior-free, the approximation component is established for every profile $\vec{s}$.
To demonstrate monotonicity in randomized mechanisms, we need to prove that monotonicity holds for any allocation that might be realized after the random choices of the mechanism.

In the construction of monotone allocations, we use the term an allocation {\em propagation}.
Suppose we assign $x(\vec{s}) = i$. Then, to ensure monotonicity of $i$'s allocation as his signal increases, we \emph{propagate} the allocation to $i$ to all profiles $s_i' > s_i$.  That is, the propagation operation is the assignment of $x(s_i', \vec{s}_{-i}) = i$ for all $s_i' > s_i$ whenever $x(\vec{s}) = i$.

A \emph{propagation conflict}, as depicted in Figure~\ref{fig:propconflict}, refers to the event where we allocate in a way such that multiple propagations cause multiple conflicting winners at some signal profile.  Let $s_i' > s_i$ and $s_j' > s_j$.  If we assign $x(s_i, s_j', \vec{s}_{-ij}) = i$, this must propagate such that $x(s_i', s_j', \vec{s}_{-ij}) = i$.  If we also assign $x(s_i', s_j, \vec{s}_{-ij}) = j$, this allocation must propagate such that $x(s_i', s_j', \vec{s}_{-ij}) = j$.  Then at $(s_i', s_j', \vec{s}_{-ij})$, we have a \emph{propagation conflict} because both $i$ and $j$ must be the winner to satisfy monotonicity, which is not possible.

Having propagation in our algorithm ensures that the allocation is monotone, so instead of verifying monotonicity of the allocation, we only need to verify that no propagation conflicts occur.

\section{Impossibility result for randomized mechanisms without single-crossing}
\label{sec:impossibility}

Consider the case where every bidder $i$'s signal space is $S_i = \{0,1\}$, and each agent $i$ has a valuation $v_i(\vec{s})=\prod_{j\neq i} s_j$; that is, the bidder has a value $1$ if and only if every other agent has signal $1$. For every bidder $i$, for some small $\vareps>0$, $$s_i=\begin{cases}
1& \quad w.p.\  \vareps\\
0& \quad w.p. \ 1-\vareps.
\end{cases}$$
The optimal expected welfare is $1$ whenever at least $n-1$ bidders have a $1$ signal. This happens with probability $\vareps^n+n\cdot \vareps^{n-1}(1-\vareps)$. Therefore,
\begin{eqnarray}
\opt \quad = \quad\vareps^n+n\cdot \vareps^{n-1}(1-\vareps) \quad  > \quad n \vareps^{n-1}(1-\vareps).\label{eq:rand_lb_opt}
\end{eqnarray}

Consider any truthful mechanism at profile $(s_i=0,\vec{s}_{-i}=\vec{1})$.  At this profile, the mechanism gets bidder $i$'s value in welfare with probability that he is allocated, $x_i(s_i=0,\vec{s}_{-i}=\vec{1})$, and otherwise gets zero since no other bidder has non-zero value.  By monotonicity, for every $i$, we have that $x_i(s_i=0,\vec{s}_{-i}=\vec{1})\leq x_i(\vec{1})$, and by feasibility, $\sum_i x_i(\vec{1})\leq 1$. Under any other profile (where at least two signals are $0$), all agents have zero value, so welfare is zero. The expected welfare of any truthful mechanism is thus bounded by
\begin{eqnarray}
\textsc{Welfare} & = & \sum_i \Pr[s_i=0,\vec{s}_{-i}=\vec{1}]\cdot x_i(s_i=0,\vec{s}_{-i}=\vec{1}) \cdot 1 +\Pr[\vec{s}=\vec{1}]\sum_i x_i(\vec{1}) \cdot 1 \nonumber\\
& = & \sum_i \vareps^{n-1}(1-\vareps)\cdot x_i(s_i=0,\vec{s}_{-i}=\vec{1})+
\vareps^n\sum_i x_i(\vec{1})\nonumber\\
&\leq & \vareps^{n-1}(1-\vareps)\sum_i x_i(\vec{1})+
\vareps^n\sum_i x_i(\vec{1})\nonumber\\
&\leq & \vareps^{n-1}(1-\vareps)+\vareps^n\nonumber\\
& = & \vareps^{n-1}.\label{eq:rand_lb_mech}
\end{eqnarray}

Combining \eqref{eq:rand_lb_opt} with \eqref{eq:rand_lb_mech}, we get that the approximation ratio of any monotone mechanism is $\textsc{Welfare}/\opt \leq \frac{1}{n(1-\vareps)}$ which can be made arbitrarily close to $1/n$; this is the same as the welfare attained by just allocating to a random bidder.

These last examples motivate the following notion of approximate single-crossing.
\section{$c$-Approximations for Settings with 2 Bidders or 2 Signals} \label{sec:warmups}

In this section we give deterministic truthful mechanisms for two special cases.

\subsection{A $c$-Approximation Mechanism for $2$ Agents} \label{subsec:2bidders}

A very simple idea gives a $c$-approximation for two bidders, independent of the sizes of their signal spaces.  Start at the signal profile that is the origin $(0,0)$ and allocate to the bidder who has the largest value at this signal profile, say bidder $1$.  Propagate the allocation as $1$'s signal increases.  Then move to the profile $(s_1 = 0, s_2 = 1)$ and repeat this: allocate to the bidder who has the highest value here, propagate as that bidder's signal increases, and then move to the profile where the losing bidder's signal has increased by one.  This algorithm is illustrated in Figure~\ref{fig:twobiddercoloring}. \\
\\
\noindent \textbf{Two-bidder coloring:}
\begin{itemize}
\item Let $s_1 = 0$ and $s_2 = 0$.
\item While $s_1 \leq |S_1|$ or $s_2 \leq |S_2|$:
	\begin{itemize}
		\item Let $i \in \argmax _i v_i(s_1,s_2)$ and let $j \neq i$.
		\item For all $s_i' \geq s_i$, set $x(s_i', s_j) = i$.
		\item Let $s_j = s_j + 1$.
	\end{itemize}
\end{itemize}

\begin{figure}[h!]
\begin{center}
	\includegraphics[scale=.3]{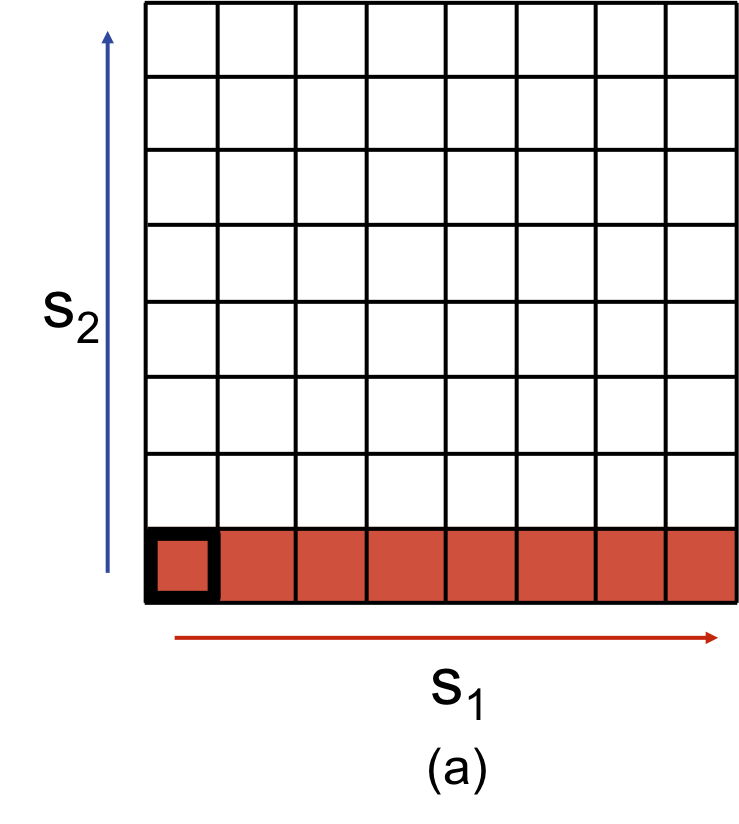} \quad\quad\quad \includegraphics[scale=.3]{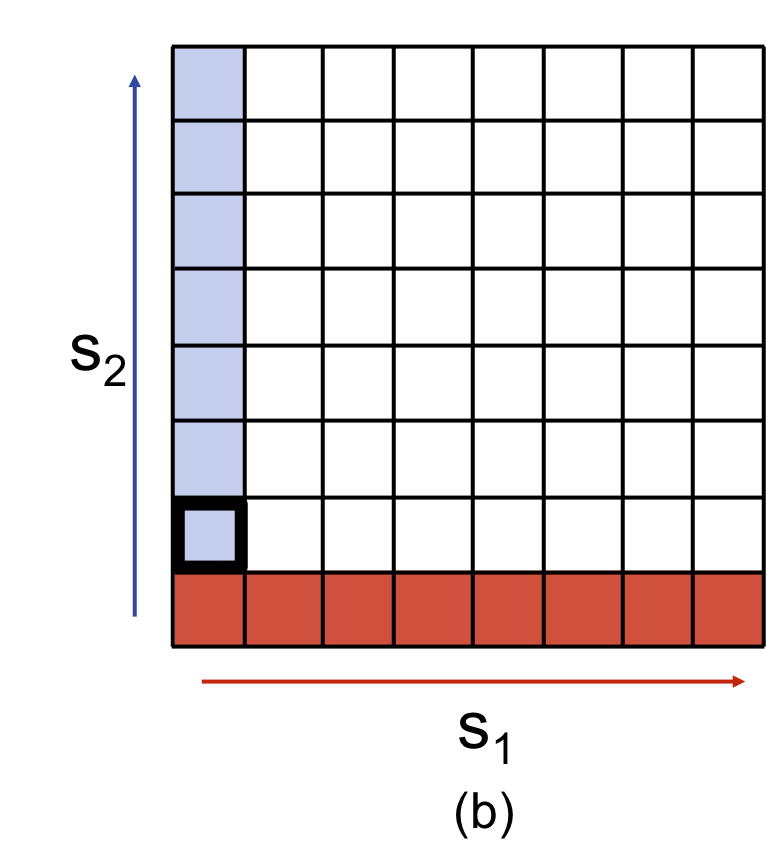} \quad\quad\quad \includegraphics[scale=.3]{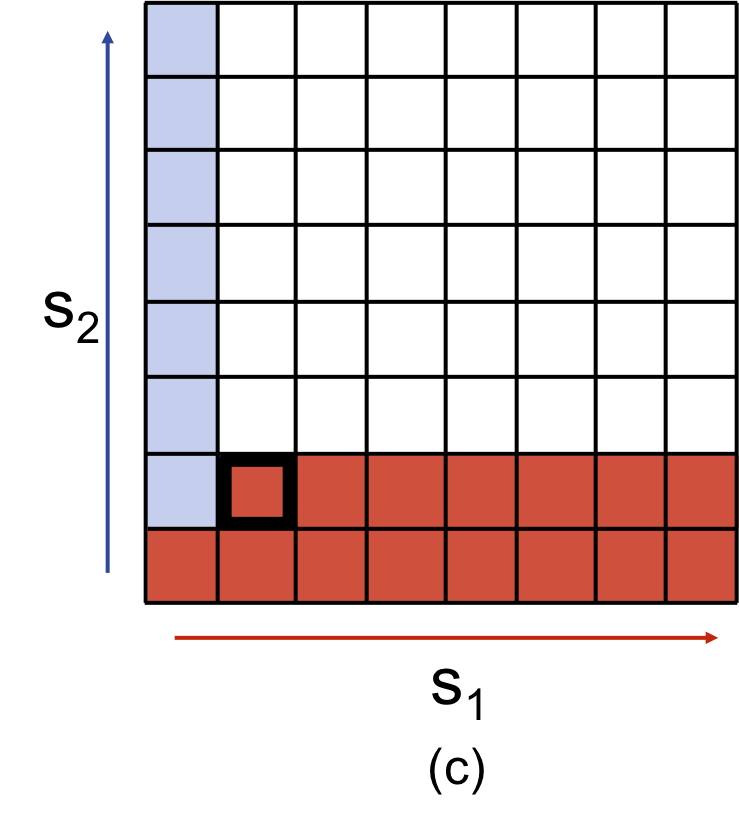}
	\caption{An illustration of the two-bidder coloring algorithm.  (a) At $(0,0)$, bidder 1 has the highest value.  Allocate to 1 and propagate this to all profiles $(s_1,0)$ for $s_1 > 0$.  (b) Move to $(0,1)$, where bidder 2 has the highest value.  Allocate to 2 and propagate this allocation to all profiles $(0, s_2)$ where $s_2 > 1$. (c) Move to $(1,1)$ where $1$ is highest.  Allocate to $1$ and propagate. } 
	\label{fig:twobiddercoloring}
\end{center}
\end{figure}

\begin{theorem}
	When we have two bidders whose valuations satisfy $c$-single-crossing, the allocation function $x$ given from the two-bidder coloring defines a truthful, deterministic, prior-free mechanism that guarantees a $c$-approximation to social welfare. \label{thm:2bidder}
\end{theorem}

\begin{proof}
(Monotonicity.) By construction, whenever we set $x(s_1,s_2)=1$, we set $x(s_1',s_2)=1$ for every $s_1'>s_1$.  Similarly, if $x(s_1,s_2)=2$, then $x(s_1,s_2')=2$ for every $s_2'>s_2$.

(Approximation.) Suppose for any profile $(s_1, s_2)$ that, without loss of generality, the item is allocated to bidder 1.  The algorithm sets $x(s_1,s_2)=1$ either because $v_1(s_1,s_2)\geq v_2(s_1,s_2)$, in which case, the highest valued bidder is the winner, or because it was true that $v_1(s'_1,s_2)\geq v_2(s'_1,s_2)$ for some $s'_1<s_1$.  By Lemma~\ref{lem:contapx}, then $c v_1(s_1, s_2) \geq v_2(s_1, s_2)$.  (A symmetric argument proves the approximation for $x(s_1,s_2)=2$.)
\end{proof}



As we will show by the example in Figure \ref{fig:c-is-tight}, the bound of $c$ is tight.  Observe that these valuations are also concave.

\subsection{A $c$-Approximation to Welfare for Settings with 2 Signals} \label{sec:2signals}

In this section, we consider the case where the size of the signal space for each bidder is at most $2$: a bidder's signal is either low ($s_i = 0$) or high ($s_i = 1$).  We will denote by $H(\vec{s})$ the set of high-signal-bidders at $\vec{s}$, that is, $H(\vec{s}) = \{i \mid s_i = 1\}$.

For monotonicity in an allocation, if we allocate to bidder $i$ at $\vec{s}$, then for all $\vec{s'} = (\vec{s}_{-i}, s'_i)$ where $s'_i > s_i$, we must propagate the allocation, allocating to bidder $i$ at these profiles as well.  The 2-signal case is special because (1) an allocation can only propagate to at most one profile, from $s_i = 0$ to $s_i$ = 1, and (2) if $s_i = 1$ and we allocate to $i$, then no propagation is necessary, so bidders with high signals are in a sense special.  We can capitalize on these properties to achieve a $c$-approximation, independent of the number of bidders $n$.


The mechanism is simple: consider profiles $\vec{s}$ in order of the number of high signals (equivalently, in increasing hamming distance $||\vec{s}||_0$ from the origin).
For any profile $\vec{s}$ in which the allocation has not already been determined, if the high-signal bidder with the largest value $i_H$ $c$-approximates the bidder with the largest value $i^*$, then allocate to $i_H$.  Otherwise, allocate to $i^*$ and propagate. \\
\\
\noindent \textbf{High-if-possible:}
\begin{itemize}
\item For all profiles $\vec{s}$ increasing in $||\vec{s}||_0$, if $x(\vec{s})$ is undefined:
	\begin{itemize}
	\item Let $i_H \in \argmax_{i \in H(\vec{s})} v_i(\vec{s})$ and $i^* \in \argmax_{i} v_i(\vec{s})$.
	\item If $v_{i_H}(\vec{s}) \capx v_{i^*}(\vec{s})$, set $x(\vec{s}) = i_H$.
	\item Otherwise, set $x(\vec{s}) = i^*$ and propagate to $x(s_{i^*} = 1, \vec{s}_{-i^*})$.
	\end{itemize}
\end{itemize}


\begin{theorem} When we have $n$ bidders where the size of the signal space is at most 2 and the valuations satisfy $c$-single-crossing, the allocation $x$ from high-if-possible gives a monotone, deterministic, and prior-free $c$-approximation to social welfare. \label{thm:2signals} \end{theorem}

\begin{figure}
\begin{minipage}[t]{0.4\textwidth}
\centering
\includegraphics[width=.7\linewidth]{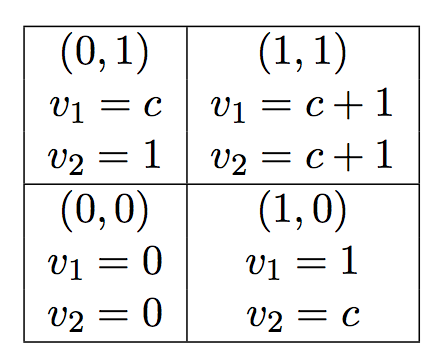}
\caption{An example with 2 bidders and 2 signals for each, where no deterministic mechanism can improve upon a $c$ approximation.} \label{fig:c-is-tight}
\end{minipage}
\hfill
\begin{minipage}[t]{0.55\textwidth}
\centering
\includegraphics[width=1\linewidth]{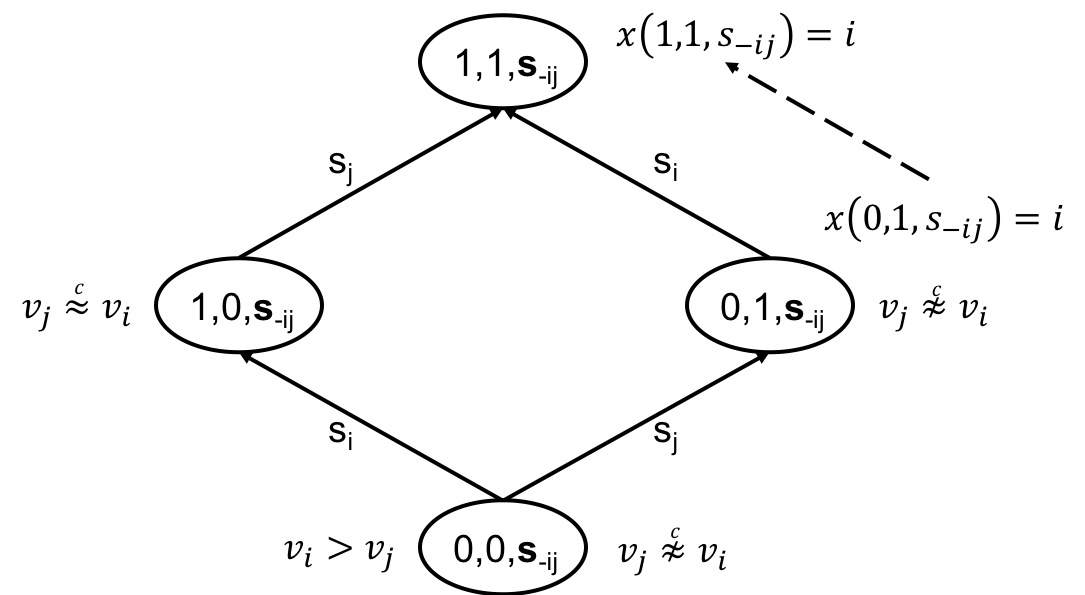}
\caption{An argument for why the algorithm results in no conflicts via the common ancestor of any two nodes that might conflict.} \label{fig:commonanestor}
\end{minipage}
\end{figure}

\begin{proof}
(Approximation.) At every profile $\vec{s}$, the allocated bidder $c$-approximates the highest valued bidder.  If the allocation at $\vec{s}$ is to bidder $i$ and is determined by propagation, by the allocation rule, it must be that at $\vec{s'} = (s_i = 0, \vec{s}_{-i})$ that $i$ was the highest valued bidder, hence $v_i(\vec{s'}) \geq v_j(\vec{s'})$ for all bidders $j$.  By 
Lemma~\ref{lem:contapx}, since only $s_i$ increases from $\vec{s'}$ to $\vec{s}$, then $c \cdot v_i(\vec{s}) \geq v_j(\vec{s})$, and thus the maximum value is $c$-approximated by $i$'s value.  Otherwise, by definition of the allocation rule, we allocate to a bidder who $c$-approximates the maximum valued bidder at $\vec{s}$.

(Monotonicity.) Our propagation step ensures monotonicity; we just need to verify that we do not cause any propagation conflicts.  The argument as to why not follows, and is illustrated in in Figure~\ref{fig:commonanestor}.  Suppose that at $\vec{s}$ we allocate to $i$ with $s_i = 0$ and propagate $i$'s allocation to $(s_i = 1, \vec{s}_{-i})$.  To have a propagation conflict at $(s_i = 1, \vec{s}_{-i})$ due to propagation of both $i$ and some other winner $j$, we observe that a propagation of $j$ can only come from allocation at a profile where $s_j = 0$ to a cell where $s_j = 1$: there could only be a conflict at $(\vec{s}_{-i}, s_i = 1)$ if $s_j = 1$.  By definition of the allocation rule, we know that $v_j(\vec{s}) \not \capx v_i(\vec{s})$, or we would have allocated to $j$ at $\vec{s}$.  Hence, by Corollary~\ref{cor:contapx},
$v_j(s_j=0, \vec{s}_{-j}) \not \capx v_i(s_j=0, \vec{s}_{-j})$, and thus $v_i(s_j = 0, \vec{s}_{-j}) > v_j(s_j = 0, \vec{s}_{-j})$.  By Lemma~\ref{lem:contapx}, this implies that 
$v_i( s_i = 1, s_j = 0, \vec{s}_{-ij}) \capx v_j( s_i = 1, s_j = 0, \vec{s}_{-ij})$.  Hence, there exists a high-signal bidder whose value $c$-approximates $j$'s value, so the algorithm would never allocate to $j$ at this profile.  Thus, there are no propagation conflicts.
\end{proof}


We next show that the $c$-approximation is tight.
In the example depicted in Figure \ref{fig:c-is-tight}, the valuations satisfy $c$-single-crossing, and any allocation $x$ that achieves better than a $c$-approximation must allocate to bidder 1 at profile $(0,1)$ and to bidder 2 at profile $(1,0)$.  However, for monotonicity, we need to propagate these allocations, which would cause a propagation conflict at $(1,1)$; hence no monotone allocation that achieves better than a $c$-approximation is possible. In fact, in Subsection~\ref{sec:random_c_lb}, we show that the $c$ is tight even if we consider random and truthful-in-expectation mechanisms that are given the prior distribution over signals.

Further, the approximations for this $2$-signal case and the $2$-bidder case of Section~\ref{subsec:2bidders} are tight in another sense, as proven in Subsection~\ref{sec:nocapprox}.


\subsection{Randomized $c$ lower bound with $c$-single-crossing}\label{sec:random_c_lb}
In Section \ref{sec:2signals}, we presented a deterministic, prior-free, and universally truthful mechanism guaranteeing a $c$-approximation for the case where each bidder has two signals. In this section, we show this is essentially tight, even if one considers randomized, truthful in expectation mechanism that are given the prior. This essentially shows that the algorithm in  Section \ref{sec:2signals} is optimal for that case.

Consider the case where there are $n$ bidders, and for $i\in\{1,\ldots,n\}$, agent $i$'s value is defined as follows:
\begin{eqnarray*}
v_i=\begin{cases}
	0\quad &s_i=0\text{ and }\exists j\neq i: s_j=0\\
	\frac{1}{c}\quad &s_i=1\text{ and }\exists j\neq i: s_j=0\\
	1\quad &s_i=0\text{ and }\forall j\neq i: s_j=1\\
	1+\frac{1}{c}\quad &s_i=1\text{ and }\forall j\neq i: s_j=1\\
\end{cases}
\end{eqnarray*}
It is easy to see the valuations are $c$-single-crossing since whenever $i$'s signal increases from $0$ to $1$, his value increases by $1/c$, while every other bidder's value increases by at most $1$. Consider the following (known) correlated distribution over signal profile.  The probability that $\vec{s}=\vec{1}$ is an arbitrarily small $\epsilon$, and the probability for any vector profile where $s_i=0$ for exactly one coordinate $i$ and $s_j=1$ for every other coordinate is $(1-\epsilon)/n$. It is immediate that $$\opt = (1-\epsilon) + \epsilon(1+1/c) > 1.$$

Consider any truthful mechanism at profile $(s_i=0,\vec{s}_{-i}=\vec{1})$.  At this profile, the mechanism gets bidder $i$'s value in welfare with probability that he is allocated, $x_i(s_i=0,\vec{s}_{-i}=\vec{1})$, and otherwise gets $1/c$. By monotonicity, for every $i$, we have that $x_i(s_i=0,\vec{s}_{-i}=\vec{1})\leq x_i(\vec{1})$, and by feasibility, $\sum_i x_i(\vec{1})\leq 1$. We have that 
\begin{eqnarray*}
\textsc{Welfare} & \leq & \sum_i \Pr[s_i=0,\vec{s}_{-i}=\vec{1}]\cdot \left(x_i(s_i=0,\vec{s}_{-i}=\vec{1}) \cdot 1+(1-x_i(s_i=0,\vec{s}_{-i}=\vec{1}))/c\right) \nonumber\\ & &+\quad\Pr[\vec{s}=\vec{1}]\sum_i x_i(\vec{1}) \cdot (1+1/c) \\
& \leq & \sum_i (1-\epsilon)/n\cdot (x_i(s_i=0,\vec{s}_{-i}=\vec{1})+1/c)+\epsilon(1+1/c)\\
&\leq & (1-\epsilon)/n\sum_i x_i(\vec{1})+(1-\epsilon)/c + \epsilon + \epsilon/c\\
&\leq & 1/n+1/c+\epsilon.
\end{eqnarray*}

Therefore, $\textsc{Welfare}/\opt\leq 1/n+1/c+\epsilon$, which as $n$ tends to $\infty$ and as $\epsilon$ is arbitrarily small, tends to $1/c$, hence we get a lower bound that tends to $c$ as $n$ tends to $\infty$.

\subsection{No $c$-approximation exists for more than $2$ agents with more than $2$ signals}\label{sec:nocapprox}

The minimal scenario that does not fall into the regime where we have already established a tight $c$-approximation is one with three bidders, two with signal space of size 2, and one with signal space of size 3.  That is, $S_1 = \{0,1,2\}$ and $S_2 = S_3 = \{0,1\}$.
We show that this scenario generally does not admit a $c$-approximation.

\begin{proposition}\label{prop:superc}
Let $n$ be the number of bidders, and $k_i$ the size of bidder $i$'s signal space. 
For any setting with $n \geq 3$, $k_i \geq 3$ for some $i$, and $c$-single-crossing valuations, there exists an instance where no truthful deterministic prior-free allocation can achieve a $c$-approximation to welfare.
\end{proposition}
\begin{proof}
Figure~\ref{fig:nocapxconstraints} depicts a valuation structure where the signal space is $S_1\times S_2 \times S_3=\{0,1,2\}\times\{0,1\}\times\{0,1\}$. We show that in order to have a $c$-approximation to the welfare, a number of constraints on the allocation must be satisfied, which, when taken together, yield a contradiction.
Numerical values that support this structure are depicted in Appendix \ref{sec:nocapproxnumbers}.

Consider the structure depicted in Figure~\ref{fig:nocapxconstraints}. At valuation profile $\vec{s}_a=(s_1=2,s_2=0,s_3=0)$, no other bidder has a value that $c$-approximates bidder~2's value, therefore, we must allocate to bidder~2 at this profile. By propagation, we must allocate to bidder~2 at profile $\vec{s}_b=(s_1=2,s_2=1,s_3=0)$ as well. Now consider profile $\vec{s}_c=(s_1=1,s_2=1,s_3=0)$. At this valuation profile, bidder~2's value does not $c$-approximate the value of the maximum bidder. Therefore, we must allocate either to bidder~1 or bidder~3 at $\vec{s}_c$. We show that both options lead to a contradiction, proving the claim.
\begin{itemize}
	\item If we allocate to bidder~1 at $\vec{s}_c$, then by propagation, we must allocate to bidder~1 at $\vec{s}_b$ as well, but we already claimed that in order to get $c$-approximation to welfare, we must allocate to bidder~2 at $\vec{s}_b$.
	\item If we allocate to bidder~3 at $\vec{s}_c$, by propagation, we must allocate to bidder~3 at $\vec{s}_d=(s_1=1,s_2=1,s_3=1)$. Now, inspecting $\vec{s}_e=(s_1=0,s_2=1,s_3=1)$, we notice that no bidder's value $c$-approximates bidder~1's value at this profile. Allocating to bidder~1 at $\vec{s}_e$ implies that we must allocate to bidder~1 at $\vec{s}_d$, which creates a propagation conflict with bidder~3 at $\vec{s}_d$.
\end{itemize}
A visualization of the implications of this structure and the conflict that must arise due to it is depicted in Figure \ref{fig:noxapxcolored}.
\end{proof}

\begin{figure}[h!]
\begin{center}
	\includegraphics[scale=.35]{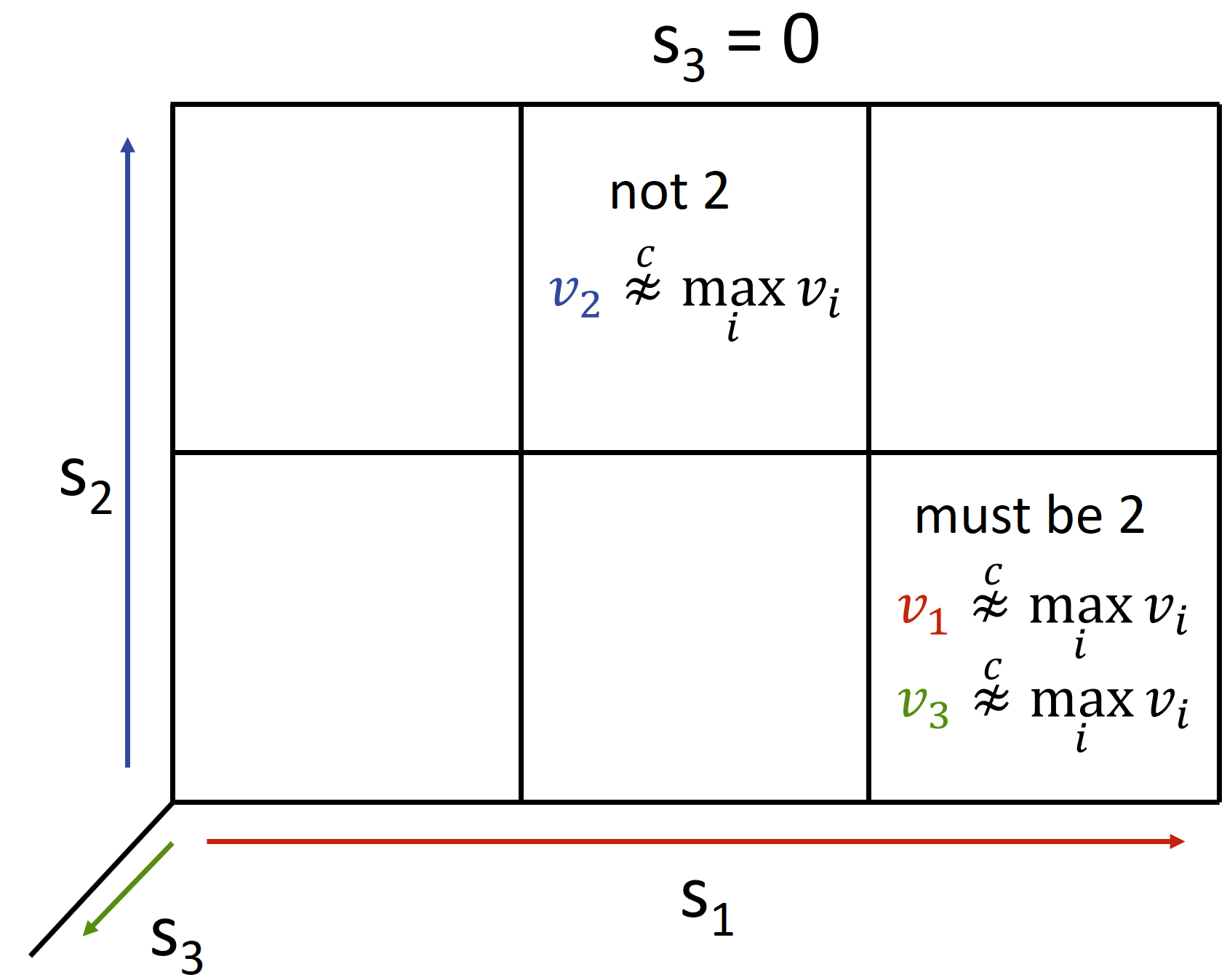} \quad\quad\quad \includegraphics[scale=.35]{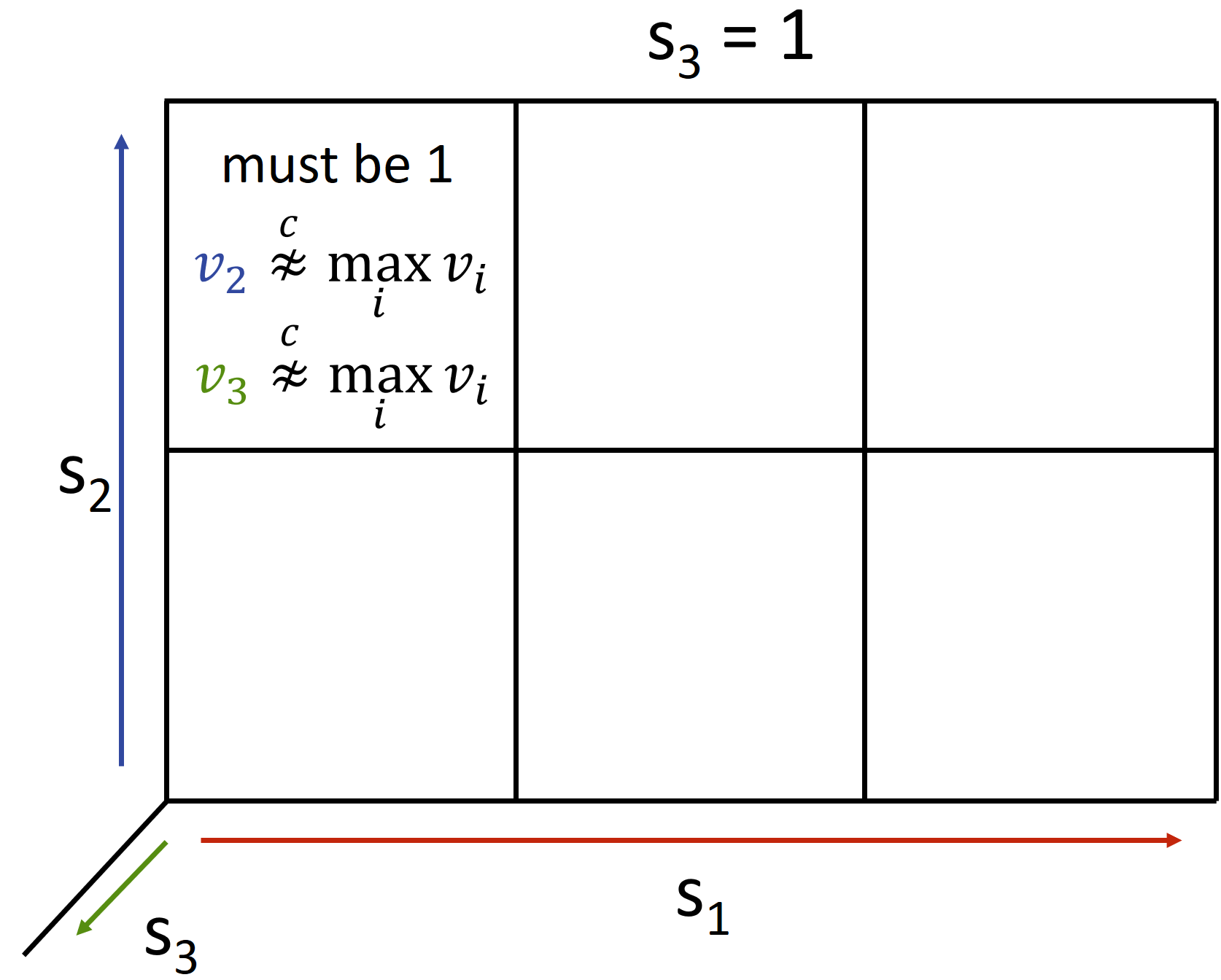}
	\caption{Left: The constraints on approximation of the values in the slice where $s_3 = 0$.  Right: The constraints where $s_3 = 1$.  The labeling ``not $i$" indicates that $i$ does not $c$-approximates the highest value.  The labeling ``must be $i$" indicates that no other bidder $c$-approximates $i$'s value.}
	\label{fig:nocapxconstraints}
\end{center}
\end{figure}

\begin{figure}[h!]
\begin{center}
	\includegraphics[scale=.35]{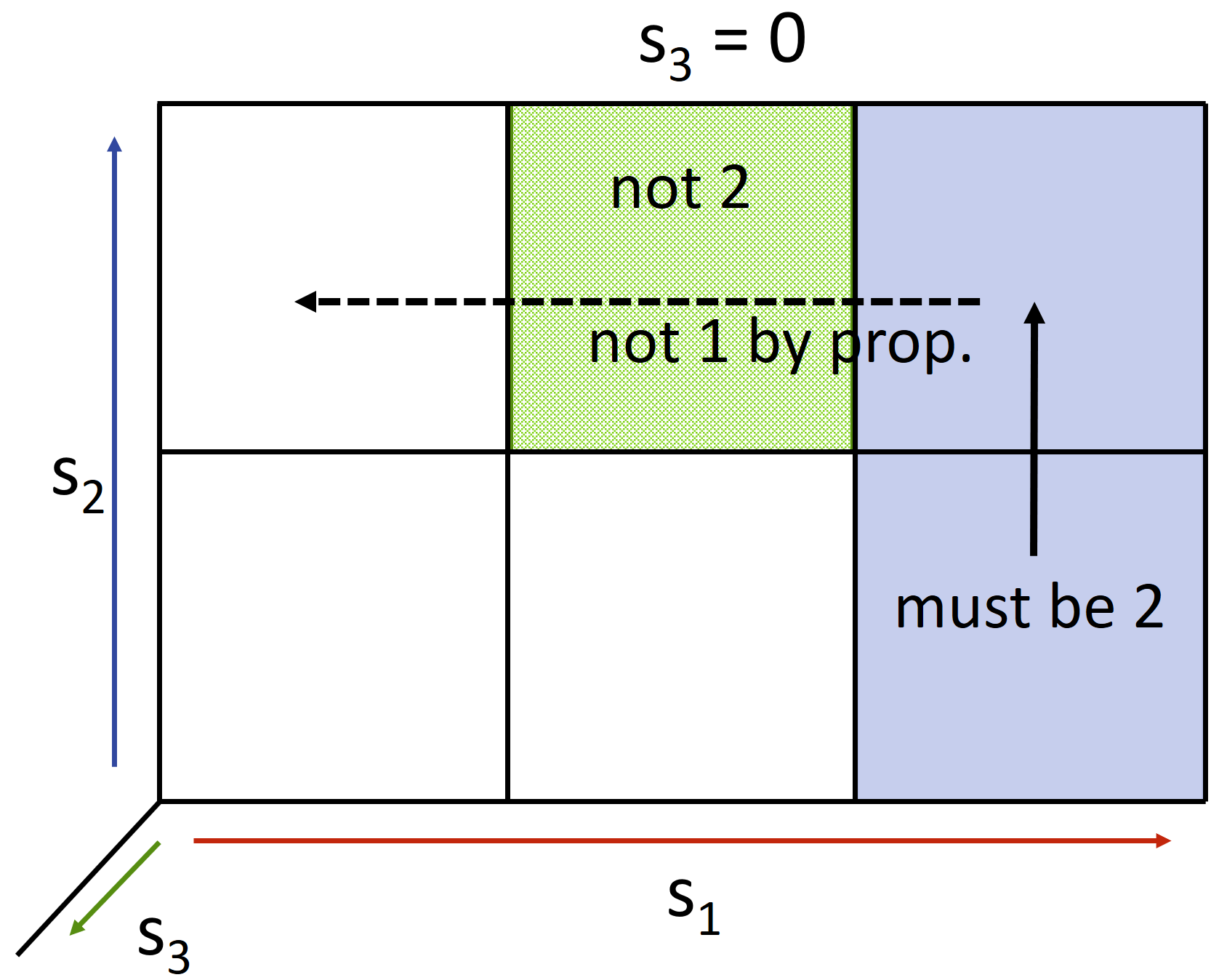}  \quad\quad \quad \includegraphics[scale=.35]{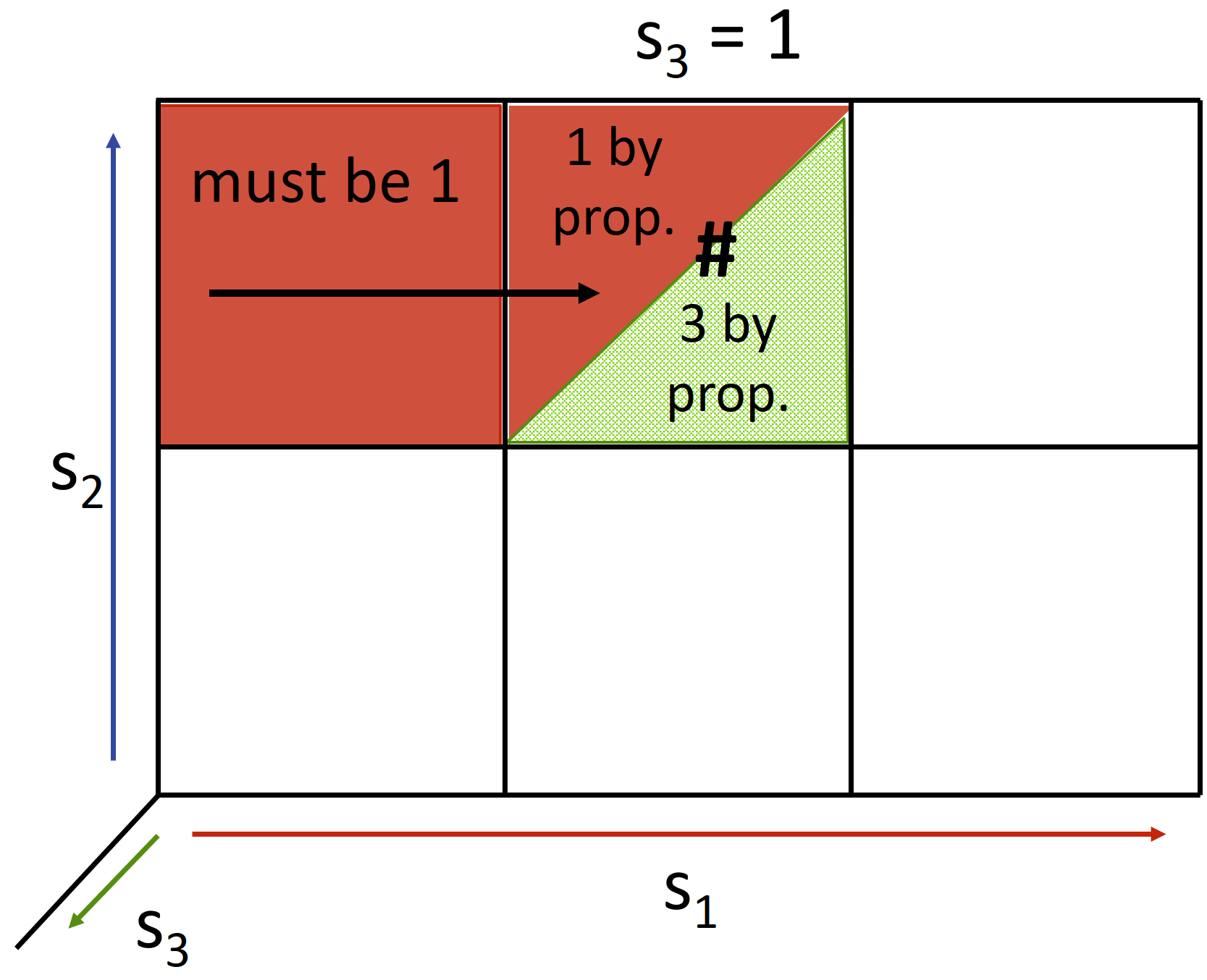}
	\caption{The implications of the valuations on the allocation rule in order to have a $c$-approximation, which leads to a propagation conflict at profile $\vec{s}_d=(1,1,1)$ where both $1$ and $3$ have been propagated as the winners.} 
	\label{fig:noxapxcolored}
\end{center}
\end{figure} 
\section{A Deterministic Mechanism for Welfare in the General Setting}
\label{sec:nbidders}

In the previous section, we have seen how in two special cases, we can achieve a truthful, deterministic, prior-free $c$-approximation to welfare, which is also tight.  However, we have also seen that there is no hope for achieving a truthful, deterministic, prior-free $c$-approximation in any setting that is more general than these.  In this section, we explore what welfare guarantee we can achieve in the general setting of $n$ bidders with any numbers of signals when using deterministic and prior-free mechanisms.


\subsection{Potential Propagation Conflicts}


The structure that arises in the 3-bidder example and makes a $c$-approximation impossible (Subsection~\ref{sec:nocapprox}) seems inevitable for any $\alpha$-approximation when $\alpha < c^2$.  What is to stop us from requiring that, in order to achieve an $\alpha$-approximation, we must \emph{not} allocate to bidder 3 at profile $(0,2,0)$ but also we \emph{must} allocate to bidder 3 at profile $(1,2,0)$?  This will force us to propagate the allocation for all $s_3' > 0$, but leave a gap at $(0,2,s_3')$ where we can \emph{never} allocate to bidder 1 without causing a conflict!  As shown in Figure~\ref{fig:towerofthrees} below, it is the case that this structure could arise.

To prevent this scenario, we want to avoid allocating to bidder $3$ at $(1,2,0)$.  
That is, we want the allocation to bidder $3$ to have the property that whenever we allocate to 3, we have also allocated to 3 at every profile where $s_1$ and/or $s_2$ are smaller.
If this is the case, then we never have this propagation in the middle of the $s_1 \times s_2$ plane.  If we only aim for a $c^2$-approximation, this is indeed possible.  If bidder 3 is the highest valued bidder at some profile $\vec{s}$, and neither bidders 1 nor 2 $c$-approximate his value, then by Corollary~\ref{cor:contapx}, the structure propagates in both dimensions toward the axes, as desired.  Otherwise, one of bidder $1$ and $2$ \emph{does} $c$-approximate 3, say bidder 1.  It turns out that we may need to allocate to bidder 2 to maintain monotonicity; however, we know that a monotone $c$-approximation to bidders 1 and 2 is possible by Theorem~\ref{thm:2bidder}, and thus, if we must allocate to bidder 2, then bidder 2 $c$-approximates bidder 1, so bidder 2 is a $c^2$-approximation to bidder 3.  For $n$ bidders, the same ideas work to give $c^{n-1}$.



What saves us from a seemingly inevitable $c^{n-1}$ factor is a subtle observation of where $c$-single-crossing continues to be helpful.  This observation, described in the ``key lemma" in the next section, shows that in an allocation like that of Figure~\ref{fig:towerofthrees}, it's actually impossible to ever \emph{need} to allocate to 1 at the profiles where allocating to 1 would cause a propagation conflict.  Thus, these propagations in the middle are actually okay, and we can design a mechanism that handles them carefully.

\begin{figure}[h!]
\begin{minipage}[t]{0.6\textwidth}
\centering
\includegraphics[scale=.2]{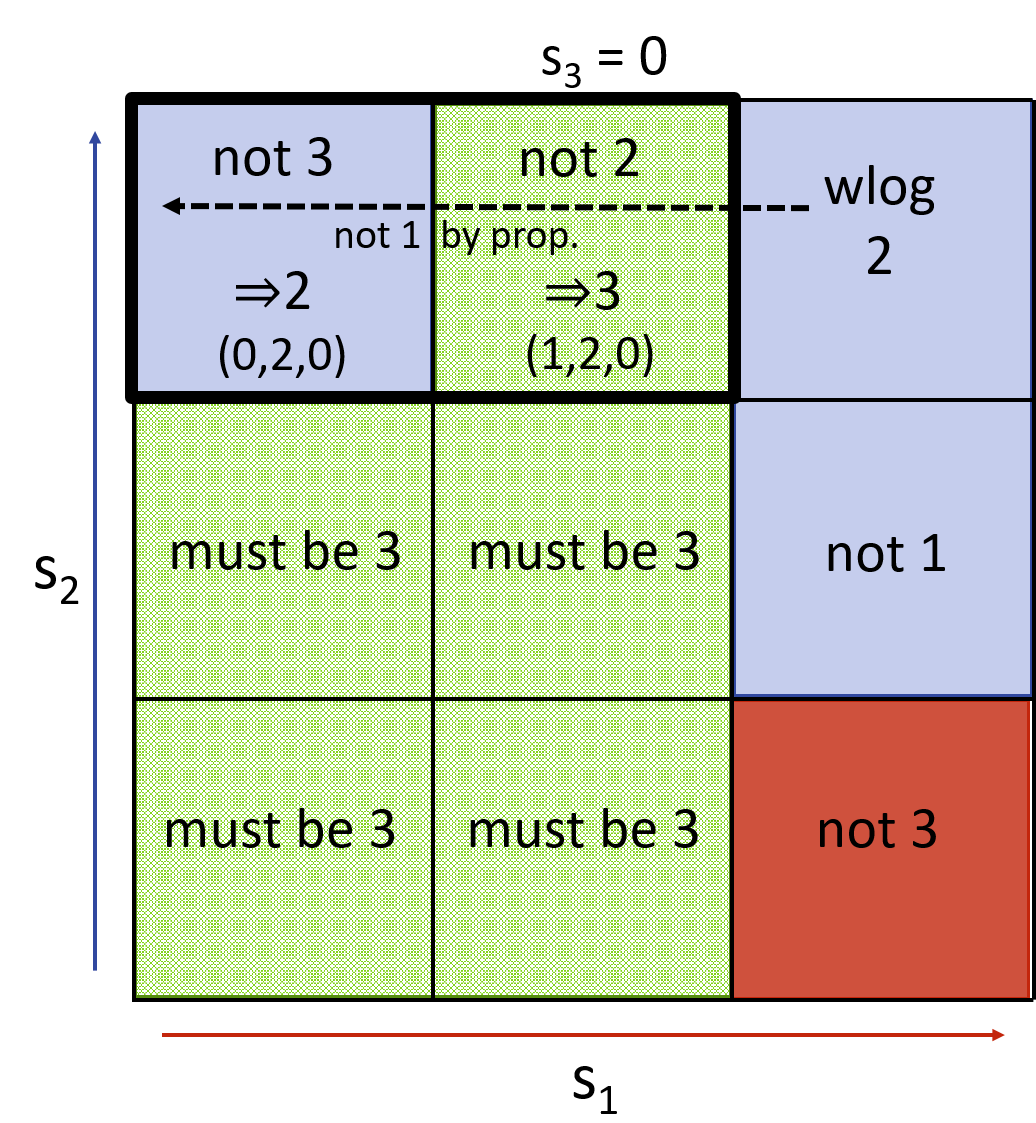} \quad\quad\quad \includegraphics[scale=.2]{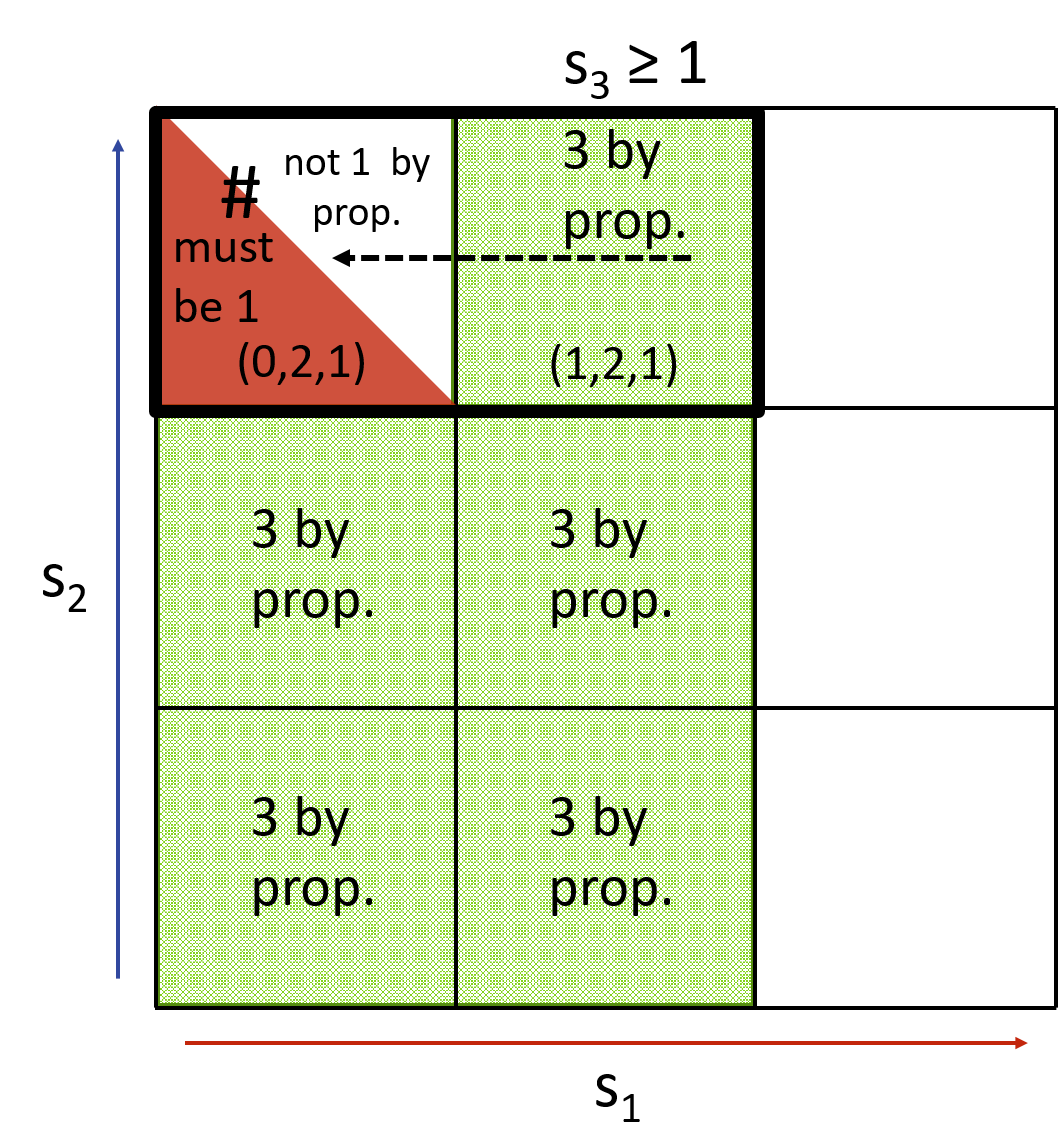}
\caption{Structures to arise due to monotonicity where certain bidders (here, bidder 1) may never be allocated to.  If the valuations ever require allocating to 1 at any profile $(0,2,s_3)$ in order to obtain an $\alpha$-approximation, it would be impossible to do so.} \label{fig:towerofthrees}
\end{minipage}
\hfill
\begin{minipage}[t]{0.35\textwidth}
\centering
\includegraphics[scale=.4]{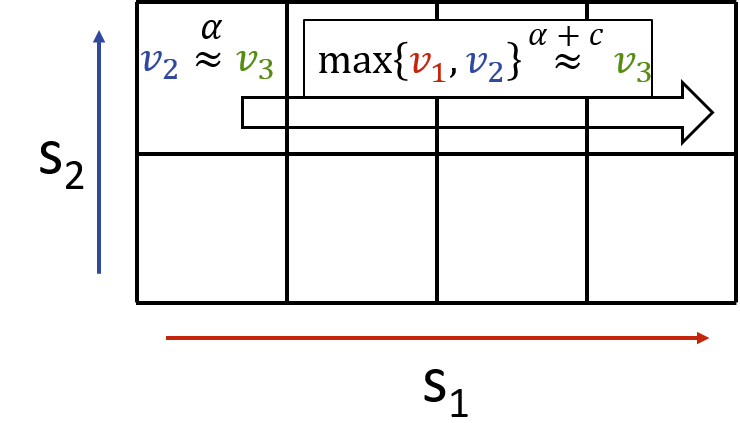}
\caption{Illustration of the key lemma: $v_2(\vec{s}) \alphaapx v_3(\vec{s})$, and as $1$'s signal increases, then $\max\{v_1, v_2\} \overset{\alpha + c}{\approx} v_3$.} \label{fig:keylemma}
\end{minipage}
\end{figure}

\subsection{The Key Lemma and Its Implications}

We observe a surprising implication of $c$-single-crossing: as we increase the signal profile in an additional dimension, we lose not a multiplicative $c$, but only an additive $c$.  See Figure~\ref{fig:keylemma}.

\begin{lemma}[Key Lemma] \label{keylemma}
Let $\vec{s}=(s_1,\ldots,s_n)$ be a signals profile. For any three agents $i,j,\ell$, if $v_i(\vec{s}) \alphaapx v_j(\vec{s})$ for some $\alpha\geq c$, then at  $\vec{s'}=(\vec{s}_{-\ell}, s'_\ell)$ where $s'_\ell\geq s_\ell$ and $\ell\neq  j$, we have that $\max\{v_i(\vec{s'}), v_\ell(\vec{s'})\}$ is an $(\alpha + c)$-approximation to $v_j(\vec{s'})$.
\end{lemma}

In English, if $v_i$ $\alpha$-approximates $v_j$, then at a profile where only $s_{\ell}$ has changed and has increased, $\max\{v_i,v_\ell\}$ $(\alpha + c)$-approximates $v_j$, so at no such profile are we required to allocate to $j$ for an $(\alpha+c)$-approximation.

\begin{proof}
By assumption, $\alpha\cdot v_i(\vec{s}) > v_j(\vec{s})$.  By $c$-single-crossing, because $\vec{s'}$ is obtained from $\vec{s}$ by increasing only $s_\ell$, then $c \left( v_\ell(\vec{s'}) - v_\ell(\vec{s}) \right) \geq v_j(\vec{s'}) - v_j(\vec{s})$.  We get that
\begin{align*}
(\alpha + c) \max\{v_i(\vec{s'}), v_\ell(\vec{s'}) \} &\geq \alpha v_i(\vec{s}) + c v_\ell(\vec{s'}) & \text{monotonicity of $v_i$}\\
&\geq \alpha v_i(\vec{s}) + c \left( v_\ell(\vec{s'}) - v_\ell(\vec{s}) \right) & \text{non-negativity of $v_\ell$}\\
&\geq v_j(\vec{s}) + v_j(\vec{s'}) - v_j(\vec{s}) & \text{adding above equations} \\
&= v_j(\vec{s'}).
\end{align*}
\end{proof}

\subsection{An $(n-1)c$-approximation for $n$ bidders}

Let $\pi = (\pi_1, \pi_2, \ldots, \pi_n)$ be an ordering over the $n$ agents, and let $k$ be the size of the signal space.

The following is the high level idea behind the algorithm that gives the allocation---we later define what it means for a value to be ``sufficiently well-approximated".  For clarity of exposition, we rename bidders so that $\pi_i$ becomes $i$.

We begin by tentatively allocating to bidder 1 at all profiles where 1 reports any signal $s_1$ and all other bidders report $0$.  Now, for bidders $j=2,\ldots,n$ (called  iteration $j$):
 \begin{enumerate} \item We correct previously determined allocations (all of which had $s_j=0$) to ensure that $j$'s value is sufficiently well-approximated (where $v_j$ is not well-approximated if her value is greater than some yet-to-be-determined factor of $\max(v_1,\ldots,v_{j-1})$).
 \item Then, we iteratively consider signal profiles as $s_j$ increases from $1$ to $k$, copying the allocation from profiles where $j$'s signal was $s_j-1$. However, we switch the allocation to bidder $j$ if the previous allocation (for the profile with smaller  $s_j$)  does not approximate bidders $1,\ldots,j$ sufficiently well.
     \end{enumerate}
      Figure~\ref{fig:hypergridcoloring} may be helpful in illustrating these ideas.

 We remark that the this algorithm, $\hypercol(\pi)$, as described, computes the allocation for every possible profile of signals. As the number of such profiles is $(k+1)^n$, it trivially follows that the time complexity is $\Omega(k^n)$. In section \ref{sec:polytime_implementation} we show how to determine the correct allocation, given an arbitrary signal profile, in polytime.

\begin{definition} \label{def:ithintprof} Given a profile of signals $\vec{s}$, a permutation $\pi$, and some index $1\leq i \leq  n$, we define the $i^{\rm th}$ intermediate profile, $\vec{s^{i}_{\pi}}$ as follows:  \begin{equation*}(\vec{s^{i}_{\pi}})_j = \begin{cases}
s_{j} &\text{if $j=\pi_\ell$ for some $1 \leq \ell \leq i$}\\
0 &\text{otherwise}
\end{cases}.\end{equation*} \end{definition}

For example, take $n=5$, $\pi=(5,2,3,1,4)$, $s=(s_1,s_2,\ldots,s_5)$, then $\vec{s^{3}_{\pi}}=(0,s_2,s_3,0,s_5)$.

In the description of the algorithm and {\sl the rest of this section} we assume that $\pi$ is the identity permutation ({\sl i.e.}, $\pi_i=i$---simply rename the bidders). We use $\vec{0}_{[j:\ell]}$ to denote $0$ values in coordinates $j, j+1, \ldots, \ell$ of the vector. For example, $(s_1,s_2, \ldots, s_{i-1}, \vec{0}_{[i:n]})$ denotes a signal profile with signals $s_1$ through $s_{i-1}$ in positions $1$ through $i-1$ and zeros elsewhere.

Note that if $\pi$ is the identity permutation we have that
$$\vec{s^{i}_\pi}= (s_1,s_2, \ldots, s_i, \vec{0}_{[i+1:n]}).$$
\\
We define the set of all possible signal profiles for signals for bidders $\ell \in \{1, \ldots, j-1\}$: $$S_{<j} = \left\{(\tilde{s}_1,\ldots,\tilde{s}_{j-1}) \mid \tilde{s}_{\ell} \in \{0, \ldots, k\}, \ell \in \{1, \ldots, j-1\}\right\}.$$
\\
\noindent \textbf{$\hypercol(\pi)$:}
\begin{itemize}
\item For all $s_1 \in \{0, \ldots, k\}$, at all profiles $(s_1, \vec{0}_{[2:n]})$, tentatively allocate to bidder~1: set $x(s_1, \vec{0}_{[2:n]}) = 1$.
\item For $j = 2, 3, \ldots, n$:
	\begin{itemize}
	\item For all $s_j \in \{0, \ldots, k\}$ and for all $\vec{s}_{<j} \in S_{<j}$:
		\begin{itemize}
		\item If $s_j > 0$ (i.e. the allocation at $\vec{s}$ is yet undefined), tentatively allocate to the same bidder as at 
		$(\vec{s}_{<j}, s_j-1, \vec{0}_{[j+1:n]})$: set 
		$$x(\vec{s}_{<j}, s_j, \vec{0}_{[j+1:n]}) = x(\vec{s}_{<j}, s_j-1, \vec{0}_{[j+1:n]}).$$
		\item If the winning bidder at $\vec{s}$ does not $(j-1)c$-approximate the maximum value among the first $j$ agents at $\vec{s}$ or does not $c$-approximate $j$'s value, reallocate to $j$: set $x(\vec{s}_{<j}, s_j, \vec{0}_{[j+1:n]})  = j$.
		\end{itemize}
	\end{itemize}
\end{itemize}

\begin{figure}[h!]
\begin{center}
\includegraphics[scale=.18]{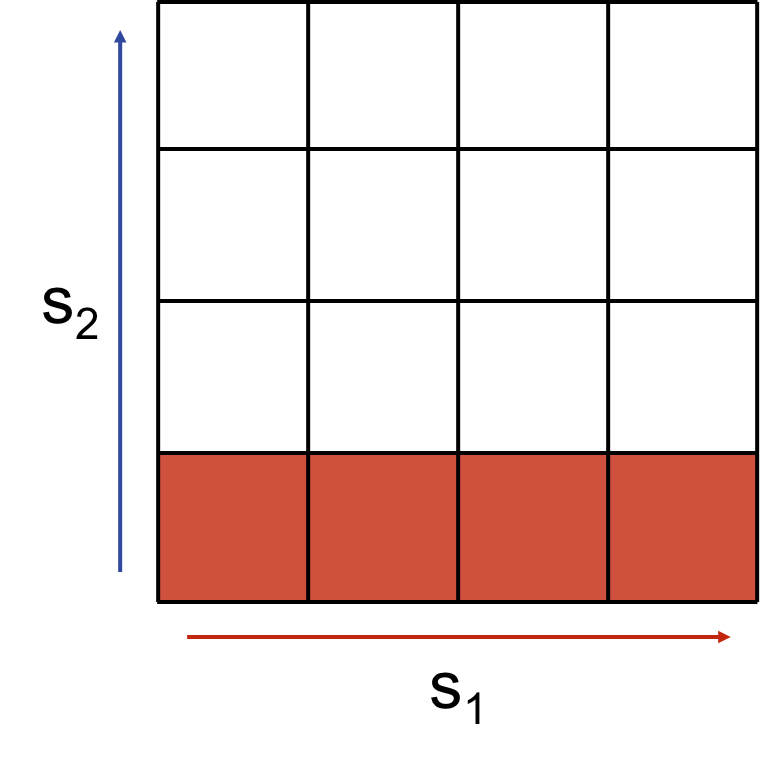} \hspace{.4cm} \includegraphics[scale=.18]{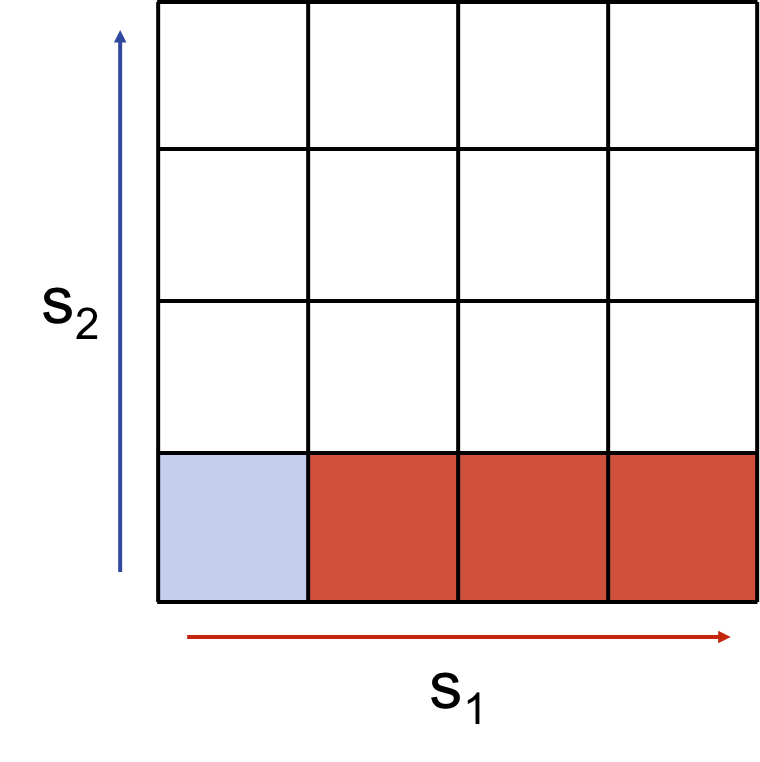} \hspace{.4cm} \includegraphics[scale=.18]{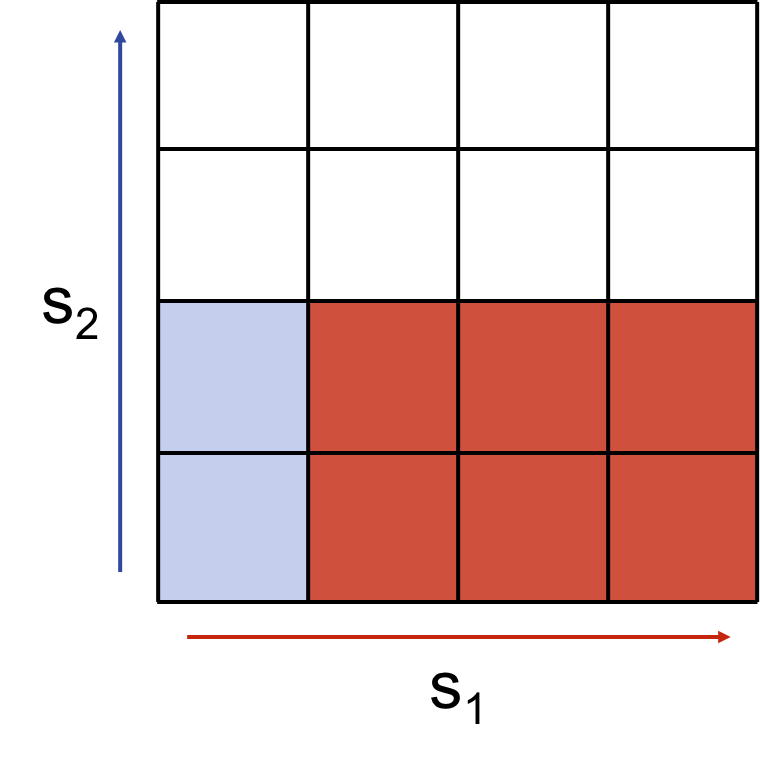} \hspace{.4cm}  \includegraphics[scale=.18]{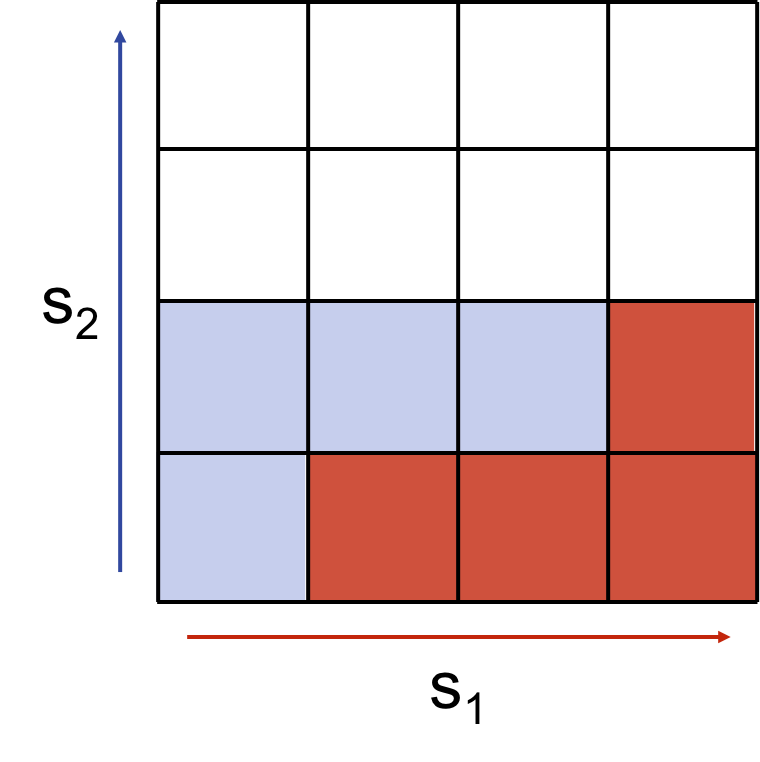} \hspace{.4cm} \includegraphics[scale=.18]{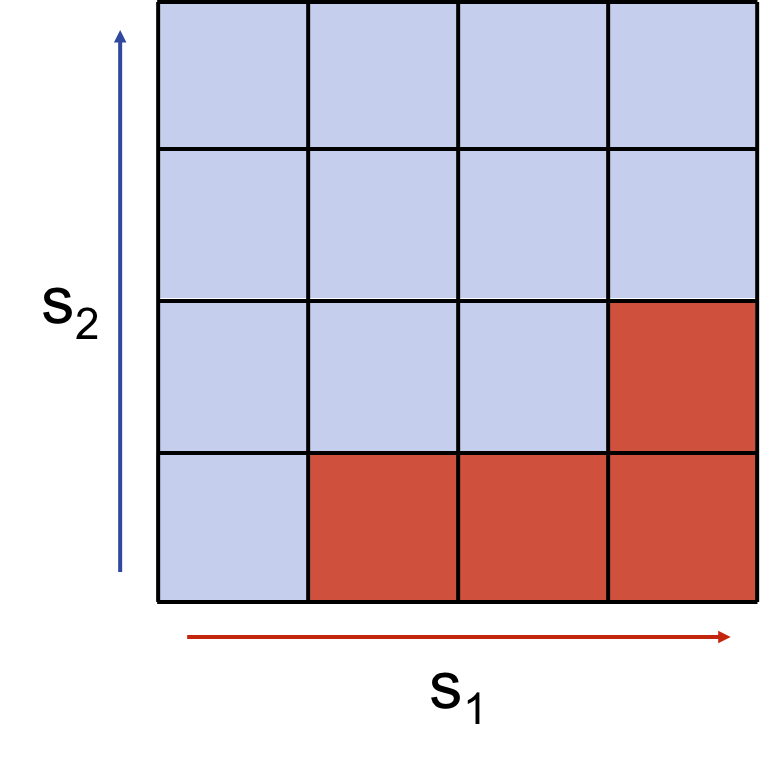}
\caption{The algorithm run on $\pi = (1, 2)$: Start by allocating to 1. Increase $s_2$ from $0$ to $k$, copying the allocation from the previous layer of $s_2$ each time and correcting to reallocate to 2 when 1 is not a sufficient approximation.
The right-most figure depicts the final allocation.}
\label{fig:hypergridcoloring}
\end{center}
\end{figure}

\begin{theorem} \label{thm:hypercolapx}
For $n$ bidders with valuations that satisfy $c$-single-crossing, the $\hypercol$ algorithm is a truthful, deterministic, and prior-free $(n-1)c$-approximation to social welfare.
\end{theorem}

\begin{proof}
(Approximation.) We show inductively that after iteration $j$, every profile $(\vec{s}_{< j+1}, \vec{0}_{[j+1:n]})$, for
$\vec{s}_{<j+1}\in S_{<j+1}$, is allocated to a bidder whose value $(j-1)c$-approximates the highest of the first $j$ bidders.


As a base case, after iteration 2, every profile $\vec{s} = (s_1, s_2, \vec{0}_{[3:n]})$, $(s_1,s_2)\in S_{<3}$,  is either allocated to bidder $1$ because $v_1(\vec{s}) \capx v_2(\vec{s})$, or the item was reallocated to bidder 2 because $v_1(\vec{s}) \notcapx v_2(\vec{s})$ and thus bidder 2's value is highest, so the item is allocated to a bidder who $c$-approximates the highest value of the first $2$ bidders.

Suppose the claim holds for iterations $1$ through $j-1$.  In iteration $j$, first we correct every profile $(\vec{s}_{<j}, 0, \vec{0}_{[j+1:n]})$, $\vec{s}_{<j}\in S_{<j}$, by reallocating to $j$ if his value is not $c$-approximated by the previous winner, ensuring that the allocation at every such profile is now a $(j-2)c$-approximation to the first $j$ bidders.

As we increase $s_j$, if the algorithm reallocates from bidder $i$ to $j$ at $(\vec{s}_{<j}, s_j, \vec{0}_{[j+1:n]})$, we know by definition of the algorithm that at profile $(\vec{s}_{<j}, 0, \vec{0}_{[j+1:n]})$, the algorithm allocated to bidder $i$, and thus $i$'s value was a $(j-2)c$-approximation to the first $j-1$ bidders.  At this profile, only $j$'s signal has increased, so by the key lemma (Lemma~\ref{keylemma}), the larger of $i$ and $j$'s values must be a $(j-1)c$-approximation to the highest value of the first $j$ bidders.  If $i$'s value is not a $(j-1)c$-approximation to all of the first $j$ bidders, then $j$'s value must be, so reallocating to bidder $j$ maintains the $(j-1)c$-approximation.

If $i$'s value \emph{was} a $(j-1)c$-approximation but we reallocated to $j$, this must be because $i$'s value wasn't a $c$-approximation to $j$'s value, so of course $j$'s value is larger.  Thus the $(j-1)c$-approximation is maintained by reallocating to $j$.  The only other possibility is that the algorithm did not reallocate to $j$, and thus $i$ passed the check that he $(j-1)c$-approximates the highest of the first $j$ bidders.

Thus, the claim holds, so after all $n$ iterations, at each profile, the item is allocated to a bidder who $(n-1)c$-approximates the highest of all $n$ bidders.


(Monotonicity.) We also show monotonicity by induction.  The first iteration is of course monotone in its allocation to bidder 1 and to no other bidder.  Suppose after iteration $j-1$, the allocation is monotone for all bidders.  In iteration $j$, we copy the same allocation from where $s_j$ is lower, which by the inductive hypothesis is monotone.  As we increase $s_j$, we only modify the allocation rule by increasing profiles where $j$ wins the item, ensuring a monotone allocation to bidder $j$ during iteration $j$.  Since we started from an allocation that was monotone to all bidders after iteration $j-1$, we only need to check that by performing our corrections where we reallocate to $j$ that we don't interrupt monotonicity for some bidder $i < j$ during iteration $j$.

Suppose that during iteration $j$ we reallocate from bidder $i$ to bidder $j$ at some profile $\vec{s}$.  By our allocation, either $i$'s value does not $c$-approximate $j$'s value, or it does not $(j-1)c$-approximate the highest value.
Then for any profile with only $s_i$ decreased, by \Cref{lem:contapx}, bidder $i$ 's value must also fail to $c$-approximate $j$ or $(j-1)c$-approximate the highest value respectively, so at all such profiles, the item will be reallocated to $j$ if it was previously allocated to $i$, as shown in Figure~\ref{fig:recolorj}.  Thus, the allocation to bidder $i$ is monotone.
\end{proof}

\begin{figure}
\begin{center}
\includegraphics[scale=.25]{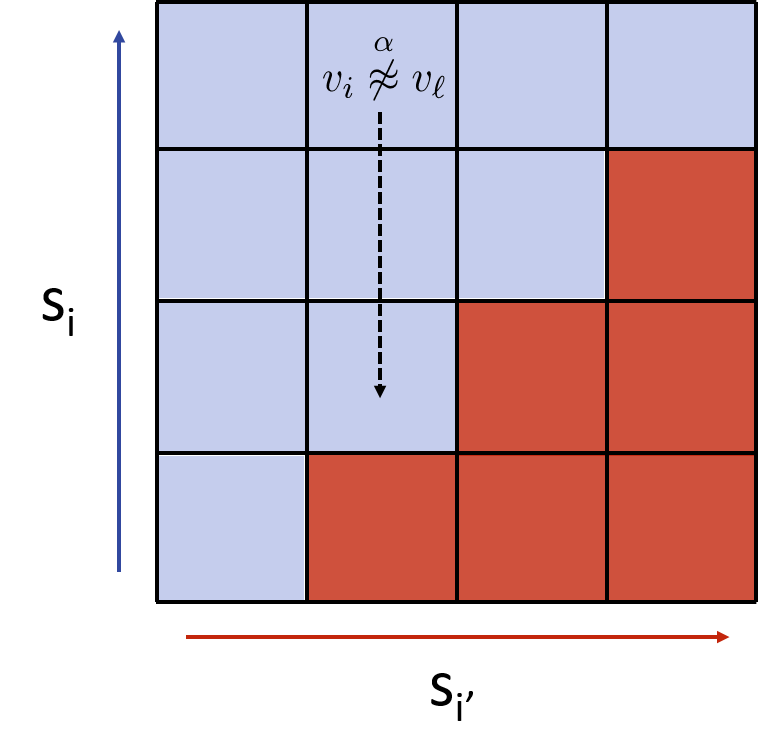} 
\includegraphics[scale=.25]{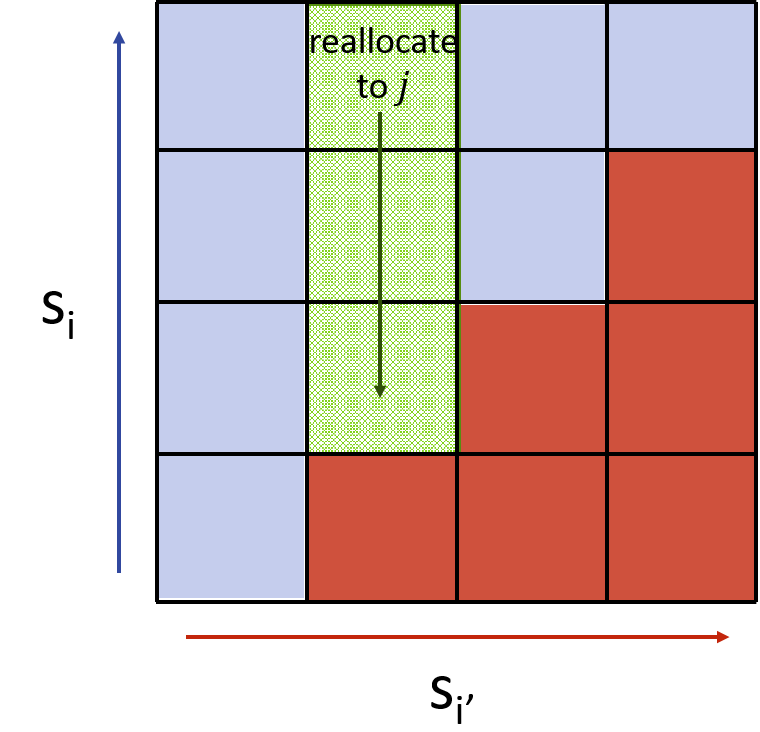}
\caption{Reallocating to bidder $j$ happens for all profiles where only $s_i$ is decreased and $i$ was previously the winner, enforcing monotonicity in the allocation to bidder $i$.}
\label{fig:recolorj}
\end{center}
\end{figure}

\subsection{Polytime Implementation for the Hypergrid-Coloring Mechanism}\label{sec:polytime_implementation}

\begin{theorem} \label{thm:polytime}
	The mechanism given by $\hypercol(\pi)$ is implementable in time $O(n^2 k\log k)$ where $n$ is the number of bidders and $k$ is the largest size of any bidder's signal space.
\end{theorem} 

\begin{proof}
First, we illustrate that given a signal profile $\vec{s}$ and ordering $\pi$, if the allocation rule $x$ is produced from $\hypercol(\pi)$, we can determine $x(\vec{s})$ in $O(n^2k)$ time.  We prove that, given the tentative winner of the $\hypercol(\pi)$ algorithm after iteration $j-1$ at profile $(\vec{s}_{<j}, \vec{0}_{[j:n]})$, we can determine the tentative winner at $(\vec{s}_{<j}, s_j, \vec{0}_{[j+1:n]})$ after the $j^{\mathrm{th}}$ iteration in $O(nk)$ time.  If $j=1$, then after the first iteration, profile $(s_1, \vec{0}_{[2:n]})$ is allocated to bidder 1 in $O(1)$ time.  

The tentative winner after iteration $j$ at profile $(\vec{s}_{<j}, s_j, \vec{0}_{[j+1:n]})$ is either the same as the tentative winner at $(\vec{s}_{<j}, 0, \vec{0}_{[j+1:n]})$ after iteration $(j-1)$, who we call bidder $i$, or the item is reallocated to $j$.  By monotonicity and the definition of $\hypercol(\pi)$, if the winner is $j$, it is because for some smallest signal $s_j' \leq s_j$, at profile $(\vec{s}_{<j}, s_j', \vec{0}_{[j+1:n]})$, bidder $i$'s value was either not a $(j-1)c$-approximation to the highest of the first $j-1$ values, or it was not a $c$-approximation to $j$'s value.  To determine whether $i$ or $j$ is the winner at $(\vec{s}_{<j}, s_j, \vec{0}_{[j+1:n]})$ after iteration $j$, we at most need to check each of the $k$ possible signals $s_j' \leq s_j$ to determine whether the item gets reallocated.  This requires determining the highest value of any of the first $j$ bidders at each such profile, which takes $O(n)$ time, thus requiring $O(nk)$ time total to compute the winner after iteration $j$.  Then computing the winner at profile $\vec{s}$ after all $n$ iterations takes at most $O(n^2k)$ time by checking the winner after iteration $j$ at profile $(\vec{s}_{\leq j}, \vec{0}_{[j+1:n]})$ for all $j = 1, \ldots, n$.

By Proposition~\ref{prop:deterministic_payment}, the payment of the winner at profile $\vec{s}$, say bidder $i$, is his value $v_i(b_i^*, \vec{s}_{-i})$ at his critical signal $b_i^*$.  In order to determine $i$'s critical signal, we can binary search over $s_i' < s_i$ to find the largest $s_i'$ where $x(s_i', \vec{s}_{-i}) \neq i$.  As there are at most $k$ signals to binary search over, this gives a factor of $\log{k}$ times the running time of determining the winner at each such profile, which we have already shown is $O(n^2 k)$.  Then determining the allocation and the payment rules runs in time $O(n^2 k \log{k})$.
\end{proof}

\subsection{Tight instance for $\hypercol$}
\label{sec:tightexample}

We give an example of valuations that satisfy $c$-single-crossing, and an ordering $\pi$ for which $\hypercol(\pi)$ produces an allocation that gives a tight prior-free $(n-1)c$-approximation to social welfare.  Consider the following valuations:
$$\forall i \neq 2 \quad v_i(\vec{s}) = \begin{cases} 0 & s_i = 0 \\ 1 & s_i = 1 \end{cases} \quad\quad\quad \text{and} \quad\quad\quad v_2(\vec{s}) = c \cdot |\{i \neq 2 \mid s_i = 1\}|. $$
These clearly satisfy $c$-single-crossing: when a bidder other than $2$'s signal increases, his value increases by 1, and bidder $2$'s signal increases by exactly $c$.  No other changes occur.

Consider the allocation $x$ produced by $\hypercol(\pi)$ where $\pi = (1,2,\cdots, n)$.  In the first iteration, $x(s_1, \vec{0}) =1 $ for all $s_1$.  Then, in iteration 2, at all profiles $\vec{s} = (s_1, s_2, \vec{0})$, $c v_1(\vec{s}) = v_2(\vec{s})$, so at no such profile will the algorithm reallocate to bidder 2, as his value is $c$-approximated.  In all iterations $j > 2$, we only reallocate to bidder $j$, thus we never allocate to bidder 2.  Then at profile $\vec{1}$, we must allocate to a bidder $i \neq 2$ with $v_i(\vec{1}) = 1$, while $v_2(\vec{1}) = (n-1)c$.  This instance gives a lower bound of $(n-1)c$ for the algorithm, illustrating that our analysis is tight. 
\section{Random $2c$-Approximation for Concave $c$-SC Valuations}
\label{sec:random}

In Section \ref{sec:tightexample} it is shown that there exists an instance and ordering $\pi$ for which $\hypercol(\pi)$ is a tight $(n-1)c$-approximation to welfare.

In this section we consider the following randomized variant:

\vspace{2mm}
\noindent\textbf{\randhypercol:}
Choose a permutation $\pi$ of the $n$ bidders uniformly at random.  Run $\hypercol(\pi)$.
\vspace{0.5mm}

We show (Theorem \ref{thm:concave-2c}) that for concave and $c$-single-crossing valuations, this mechanism gives a $2c$-approximation.
We also show that, without concavity, the same randomized mechanism gives an (inferior) $2c^{3/2} \sqrt{n}$-approximation (see Theorem \ref{thm:randrootn}, whose proof appears in Appendix~\ref{sec:random-appendix}). Nonetheless, note that $\randhypercol$ gives an asymptotic improvement in the approximation over the deterministic $\hypercol(\pi)$.  Moreover, we show that the $\sqrt{n}$ factor is unavoidable for $\randhypercol$ if valuations are not concave.

In this Section and in Section \ref{sec:random-appendix} we {\sl do not} assume that $\pi$ is the identity permutation. Of course, internally to  $\hypercol(\pi)$, such a translation can be done.

\begin{theorem} For arbitrary $c$-single-crossing valuations,
 $\randhypercol$ is a universally truthful and prior-free mechanism that gives a $2c^{3/2} \sqrt{n}$-approximation to social welfare.  \label{thm:randrootn}
\end{theorem}

\begin{theorem} For concave valuations that are also $c$-single-crossing,
 $\randhypercol$ is a universally truthful and prior-free mechanism that gives a $2c$-approximation to social welfare.  \label{thm:concave-2c}
\end{theorem}

Rather than proving Theorem \ref{thm:concave-2c}, we give a more general result from which Theorem \ref{thm:concave-2c} follows as a special case.
In particular, we consider a parameterized version of concave valuations, $d$-concave valuations, for which the special case of $d=1$ reduces to the standard definition of concave valuations.

Theorem \ref{thm:d-concave} below shows a $c(d+1)$-approximation guarantee for $d$-concave valuations.
Before giving the proof of Theorem \ref{thm:d-concave}, we give some intuition.

Recall Definition \ref{def:ithintprof} of the $i^{\rm th}$ intermediate profile $\vec{s^{i}_{\pi}}$; this is a function of a signal profile $s$ and a permutation $\pi$.

Fix a signal profile $\vec{s}$, let $v_{i^*}(\vec{s})=\max_\ell v_\ell(\vec{s})$,
and let $t$ be $i^*$'s position in the randomly drawn ordering $\pi$ ({\sl i.e.}, $\pi_t=i^*$; see Figure \ref{fig:permutation}).
Recall that if the tentative winner at $\vec{s^t_{\pi}}$, {\sl determined in the $t^{th}$ iteration}, is not $i^*$ then it must $c$-approximate $i^*$'s value at $\vec{s^t_{\pi}}$. Alternatively, the item is reallocated to $i^*$. Therefore, the expected value of the tentative winner at $\vec{s^t_{\pi}}$ is at least $1/c$ of the expected value of $i^*$ at $\vec{s^t_{\pi}}$.
It now follows from Lemma~\ref{boundedalgchange} that this also holds for the expected value of the final winner at $\vec{s}$.
Therefore, it only remains to show that $\E_{\pi}\left[v_{i^*}(\vec{s^t_{\pi}})\right]\geq v_{i^*}(\vec{s})/2.$

We say that the signal of bidder $i$ ``turns on" when it changes from $0$ to $s_i$, and ``turns off" when it changes from $s_i$ to $0$.
Consider some permutation $\pi$ where the position of $i^*$ is $t$.  Let $A\subset \{1,\ldots,n\}$ be the set of bidders preceding $i^*$ in $\pi$, and let $B=\{1,\ldots,n\}\setminus (A\cup i^*)$ be the set of bidders succeeding $i^*$. If the value of $i^*$ after the signals from bidders in $A$ and $i^*$ have turned on is $v_{i^*}(\vec{s^t_{\pi}})=\alpha\cdot  v_{i^*}(\vec{s})$ for some $0\leq \alpha\leq 1$, then obviously, the change in value of $i^*$ when we turn on all signals of $B$ is exactly $(1-\alpha)\cdot  v_{i^*}(\vec{s})$. Now consider a permutation $\pi'$ that is obtained from $\pi$ by switching between the bidders in $B$ and the bidders in $A$. That is, in $\pi'$, the bidders in $B$ precede $i^*$ and the bidders in $A$ succeed $i^*$.
Let $t'$ be the position of bidder $i^*$ in $\pi'$ (see Figure \ref{fig:permutation} for an illustration).

\begin{figure}[h!]
\begin{center}
\includegraphics[scale=.3]{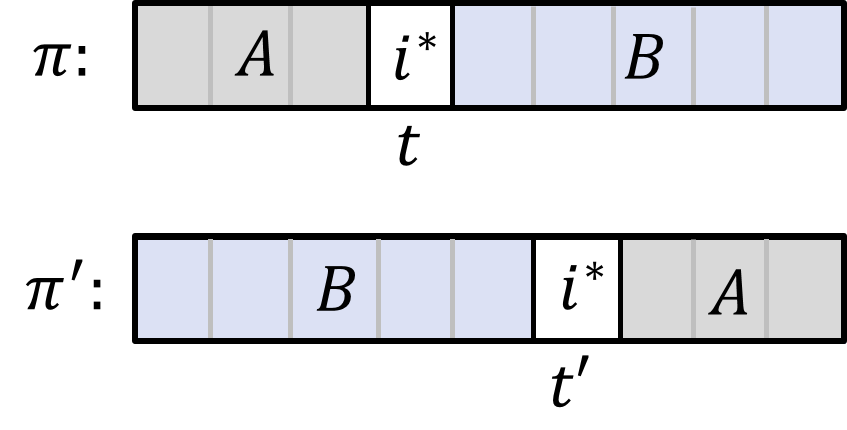}
\caption{An illustration of the permutations $\pi$ and $\pi'$.}
\label{fig:permutation}
\end{center}
\end{figure}

We now argue that 
$$v_{i^*}(\vec{s^{t'}_{\pi'}})\geq (1-\alpha)\cdot v_{i^*}(\vec{s}).$$
To see this, note that the increase in $v_{i^*}$ as a result of the signals of the bidders in $B$ turning on, given that the signals of the bidders in $A$ are on, is $(1-\alpha)\cdot  v_{i^*}(\vec{s})$.
It follows from concavity that the increase in $v_{i^*}$ as a result of the signals of the bidders in $B$ turning on, given that the signals of the bidders in $A$ are off, can only be larger.
To conclude the argument, since $\pi$ and $\pi'$ are drawn with equal probabilities, and $v_{i^*}(\vec{s^t_{\pi}})+v_{i^*}(\vec{s^{t'}_{\pi'}})\geq v_{i^*}(\vec{s})$, it holds that $\E_{\pi}\left[v_{i^*}(\vec{s^t_{\pi}})\right]\geq v_{i^*}(\vec{s})/2$, as desired.

We now state the more general Theorem and give the formal proof.

\begin{theorem}
For every $d$-concave and $c$-single-crossing valuations, $\randhypercol$ is a randomized universally truthful and prior-free mechanism that gives a $c(d+1)$-approximation to social welfare.
\label{thm:d-concave}
\end{theorem}

\begin{proof}
Given a profile $\vec{s}$, let $i^*$ be the highest valued bidder in $\vec{s}$; i.e., $i^* \in \argmax_i v_i(\vec{s})$.
We show that the expected value of $i^*$ at the iteration where $i^*$ is considered in $\pi$ is at least $\frac{1}{d+1}$ of his value at $\vec{s}$.
That is, it holds that $\E_\pi\left[v_{i^*}(\vec{s_\pi^{t}})\right] \geq  \frac{v_{i^*}(\vec{s})}{d+1}$, where $t$ is the random variable representing $i^*$'s location in a random permutation $\pi$.

Before proving this claim, we show how it implies a $c(d+1)$-approximation guarantee.
Let $j'$ be the tentative winner in $i^*$'s iteration under $\pi$.
By the properties of $\hypercol(\pi)$, $j'$'s value at $\vec{s^{t}_{\pi}}$ $c$-approximates $i^*$'s value, and by Lemma~\ref{allocatedvaluemonotonicity}, then $v_j(\vec{s}) \geq v_{j'}(\vec{s_{\pi}^{t}})$, where $j$ is the final winner; i.e., $j = x(\vec{s})$.
	
Putting it all together, if $\E_\pi\left[v_{i^*}(\vec{s_{\pi}^{t}})\right] \geq \frac{1}{d+1} v_{i^*}(\vec{s})$, then
	$$c \cdot \E_\pi\left[v_j(\vec{s})\right] \quad\quad \geq \quad\quad  c \cdot \E_\pi\left[v_{j'}(\vec{s_{\pi}^{t}})\right] \quad\quad  \geq \quad\quad  \E_\pi\left[v_{i^*}(\vec{s_{\pi}^{t}})\right] \quad\quad  \geq \quad\quad  \frac{1}{d+1}  v_{i^*}(\vec{s}),$$
as desired.	

We now prove that the expected value of $i^*$ at $\vec{s_{\pi}^{t}}$ is at least $1/(1+d)$ fraction of his value at $\vec{s}$.  Fix an ordering $\pi$, and let $t$ be $i^*$'s position in $\pi$.
Let $A\subset \{1,\ldots,n\}$ (resp., $B=\{1,\ldots,n\}\setminus(A\cup i^*)$) be the set of bidders whose location in $\pi$ is smaller than (resp., greater than) $t$.
That is, $A$ and $B$ are the sets of bidders preceding and succeeding $i^*$ in $\pi$, respectively.
Let $0\leq \alpha\leq 1$ be the fraction of $i^*$'s value recovered at $\vec{s_{\pi}^{t}}$, that is
\begin{eqnarray}
	v_{i^*}(\vec{s_{\pi}^{t}})= \alpha\cdot v_{i^*}(\vec{s}).
\label{eq:intermediate_eq}
\end{eqnarray}

It follows that the increase in $i^*$'s value as a result of the signals in $B$ turning on is $(1-\alpha)\cdot v_{i^*}(\vec{s})$. Now consider the permutation $\pi'$ in which $A$'s and $B$'s internal ordering is the same as in $\pi$, but where $B$ precedes $i^*$, which in turn precedes $A$, as shown in Figure \ref{fig:permutation}. Let $t'$ be $i^*$'s position in $\pi'$.
Since at time $t'$, all the signals of bidders in $A$ have yet to turn on and are thus 0, $d$-concavity implies that the change in $i^*$'s value as a result of turning on the signals in $B$ in $\pi'$ is at least $1/d$ fraction of the change in $\pi$.
That is,
$$v_{i^*}(\vec{s_{\pi'}^{t'}}) - v_{i^*}(\vec{0}) \geq \frac{1}{d}\cdot \left( v_{i^*}(\vec{s}) - v_{i^*}(\vec{s_{\pi}^{t}}) \right) $$
which by non-negativity of $v_{i^*}(\vec{0})$ and by (\ref{eq:intermediate_eq}) implies that
	\begin{eqnarray}
		v_{i^*}(\vec{s_{\pi'}^{t'}})\geq \frac{1-\alpha}{d}\cdot v_{i^*}(\vec{s}). \label{eq:intermediate_lb1}
	\end{eqnarray}

By a similar argument, if the change in $i^*$'s value after turning on the signals from $A$ under $\pi$ is at most $\alpha\cdot v_{i^*}(\vec{s})$, then by $d$-concavity, the change in $i^*$'s value after turning on the signals from $A$ under $\pi'$ is even smaller---at most $\alpha d\cdot v_{i^*}(\vec{s})$:
$$v_{i^*}(\vec{s}) - v_{i^*}(\vec{s_{\pi'}^{t'}}) \leq d \cdot \left( v_{i^*}(\vec{s_{\pi}^{t}}) - v_{i^*}(\vec{0}) \right) \leq \alpha d\cdot v_{i^*}(\vec{s}). $$
Therefore, it holds that\footnote{Note that the right hand side of Equation ~\eqref{eq:intermediate_lb2} might be negative, but this is fine.}
	\begin{eqnarray}
	v_{i^*}(\vec{s_{\pi'}^{t'}})\geq (1-\alpha d)\cdot v_{i^*}(\vec{s}). \label{eq:intermediate_lb2}
	\end{eqnarray}	

It follows from Equations \eqref{eq:intermediate_eq} and \eqref{eq:intermediate_lb1} that
	\begin{eqnarray}
		v_{i^*}(\vec{s_{\pi}^{t}})+v_{i^*}(\vec{s_{\pi'}^{t'}})\geq \frac{1+(d-1)\alpha}{d}\cdot v_{i^*}(\vec{s}), \label{eq:sum1}
	\end{eqnarray}
and combining Equations \eqref{eq:intermediate_eq} and \eqref{eq:intermediate_lb2} gives
	\begin{eqnarray}
	v_{i^*}(\vec{s_{\pi}^{t}})+v_{i^*}(\vec{s_{\pi'}^{t'}})\geq (1-(d-1)\alpha)\cdot v_{i^*}(\vec{s}). \label{eq:sum1}
	\end{eqnarray}
Therefore, \begin{eqnarray}
		v_{i^*}(\vec{s_{\pi}^{t}})+v_{i^*}(\vec{s_{\pi'}^{t'}})\geq\max\left\{\frac{1+(d-1)\alpha}{d},1-(d-1)\alpha\right\}\cdot v_{i^*}(\vec{s}).\label{eq:lb_max}
	\end{eqnarray}
The last expression obtains its minimum when the two terms in the $\max$ expression are equal, which occurs at $\alpha=\frac{1}{d+1}$. Plugging this value of $\alpha$ into Equation \eqref{eq:lb_max} implies that for every pair of $\pi$ and its corresponding $\pi'$ (note that this is a bijection) it follows that
\begin{eqnarray*}
	v_{i^*}(\vec{s_{\pi}^{t}})+v_{i^*}(\vec{s_{\pi'}^{t'}})\geq \frac{2}{d+1}\cdot v_{i^*}(\vec{s}).\label{eq:sum_lb}
\end{eqnarray*}

Since there are $n!/2$ such pairs (half the number of permutations), we can write $$\sum_{\pi}v_{i^*}(\vec{s_{\pi}^{t}})\geq \frac{n!}{2}\cdot \frac{2}{d+1}\cdot v_{i^*}(\vec{s})=\frac{n!}{d+1}\cdot v_{i^*}(\vec{s}).$$
	We conclude that
	\begin{eqnarray*}
		\E_\pi \left[v_{i^*}(\vec{s_{\pi}^{t}})\right] & = &\sum_{\pi}\Pr[\pi]\cdot v_{i^*}(\vec{s_{\pi}^{t}})\\
		&= &\frac{1}{n!}\cdot \sum_{\pi}v_{i^*}(\vec{s_{\pi}^{t}})\\
		&\geq&\frac{1}{n!}\cdot \frac{n!}{d+1}\cdot v_{i^*}(\vec{s})\\
		&=&\frac{1}{d+1}\cdot v_{i^*}(\vec{s}),
	\end{eqnarray*}
	as desired.
\end{proof}

\section{Revenue Approximation} \label{sec:revenue}
\subsection{Black-box reduction for deterministic mechanisms}
In this section, we give a black box mechanism that, when given a deterministic truthful allocation rule $x$ that gives an $\alpha$-approximation to welfare for some $\alpha\geq1$, earns expected revenue $\frac{1}{\alpha^2 + 4\alpha + 1} \, \opt$.  The mechanism is very similar to that of \citet{CFK}. While \citet{CFK} build upon the generalized VCG mechanism, we apply similar ideas to an arbitrary deterministic monotone allocation $x$. We later show how to extend this to a family of randomized mechanisms, including the mechanism $\randhypercol$ given in Section~\ref{sec:random}.
As in \citet{CFK}, to approximate revenue we require concave valuations (see Definition~\ref{def:concavity}).

Let $x(\cdot)$ be a truthful allocation rule is defined for any number of bidders. Consider $n$ bidders, with signals $s_i$, $1 \leq i \leq n$ and
let $v_i(\vec{s})$ be the valuation of bidder $i$ on the signal profile $(s_1, \ldots, s_n)$.  Let $S\subset \{1,\ldots,n\}$ be a subset of the $n$ bidders, and let $\vec{s}_S$ denote set of signals of bidders in $S$, while $\vec{s}_{-S}$ denotes the set of signals of the other bidders. For $i\in S$ define the valuation $v^S_i(\vec{s}_S)=v_i(\vec{s}_S,\vec{s}_{-S}).$
Fix the set of signals $s_S$, define $x_S(\vec{s}_S)=x(\vec{s}_S,\vec{s}_{-S})$. If $x(\cdot)$ is a truthful allocation rule with approximation $\alpha$ with respect to some set of valuations $v_i(\vec{s})$ then $x^S(\cdot)$ is a truthful allocation rule with approximation $\alpha$ with respect to valuations $v_i^S(\vec{s}_{S})$.

\begin{definition} \label{def:winningcondreserve} Bidder $i$'s \emph{winning conditional monopoly reserve price} $p_i^*: p_i^*(\vec{s}_{-i}, x)$ is a function of the other bidder's signals, $\vec{s}_{-i}$, and the allocation rule $x(\cdot)$.  Fix all other bidder signals $\vec{s}_{-i}$, and let  $b_i^*=\min_{b_i} x_i(\vec{s}_{-i},b_i)=1$ be the critical signal for bidder $i$. Now, $p_i^*$ is the monopoly selling price given that $s_i\geq b_i^*$ and $s_{-i}$.
   Let $R_i$ be the expected revenue from offering the item to $i$ at price $p_i^*$.
\end{definition}

\begin{definition} \label{def:losingcondreserve} Bidder $i$'s \emph{losing conditional monopoly reserve price} $p_i^*: p_i^*(\vec{s}_{-i}, x)$ is a function of all of the other bidder's signals $\vec{s}_{-i}$ and the allocation rule $x(\cdot)$. Fix all other bidder signals $\vec{s}_{-i}$, and, as above, let  $b_i^*=$ be the critical signal for bidder $i$. Now, $p_i^*$ is the monopoly selling price given $s_{-i}$ and that $s_i< b_i^*$.
   Let $\tilde{R}_i$ be the expected revenue from offering the item to $i$ at this price.
  \end{definition}

Given a deterministic allocation function $x$ with welfare approximation guarantee $\alpha$, consider the following [randomized] blackbox mechanism $M$:
\begin{enumerate}
\item[(a)] With probability $\frac{\alpha^2 +1}{\alpha^2 + 4\alpha d + 1}$:  Given signals $\vec{s}$, let $i := x(\vec{s})$ be the winner of the original allocation at $\vec{s}$.  Offer bidder $i$ the item at price $p_i^*(\vec{s}_{-i}, x(\cdot))$. Note that the item need not be sold.
\item[(b)] With probability $\frac{4\alpha d}{\alpha^2 + 4\alpha d + 1}$:
Let $Z\subset \{1,\ldots,n\}$ be chosen uniformly at random from all such subsets.   Let $i := x^{Z}(\vec{s}_Z)$ be the winner of the allocation restricted to potential winners from $Z$.  Give the item to $i$ if her value at the reported signals is above his conditional monopoly reserve price $p_i^*(\vec{s}_{-i}, x(\cdot;Z))$.
\end{enumerate}

\begin{theorem} \label{thm:revapx}
Mechanism $M$ earns a $(\alpha^2 + 4\alpha d + 1)$-approximation to the optimal revenue.
\end{theorem}

\begin{proof}

Rename the bidders such that $1 = x(\vec{s}; [n])$ is the winner among all $n$ bidders, $2 \in \argmax _{i \neq 1} v_i(\vec{s})$ is a bidder with the highest value aside from $1$, and $3 = x(\vec{s}; Z)$ is the winner among the random set of potential winners $Z$.

Fix a set of signals $\vec{s}$. First we observe that when the mechanism executes (a), the expected revenue is exactly $R_1$, as the mechanism offers the winning conditional monopoly reserve price to $1$.  When the mechanism executes (b) and the winner is bidder $3$, the expected revenue is at least $\tilde{R}_3$ for any random set of bidders $Z$. This is because, conditioning that $3$'s signal is more than the critical signal $b_3^*(s_{-i}, x(\cdot;Z))$ can only increase the revenue when he is the winner at $x(\cdot;[n])$ (in which case, without conditioning the expected revenue is at least $R_3>\tilde{R}_3$) or not (in which case, without conditioning the the expected revenue is at least $\tilde{R}_3$). Notice that the events that we run (b), $1\notin Z$ and $2\in Z$ are independent. Let $A$ be the event that $1\not\in Z$ and $2\in Z$. We therefore bound the revenue of the mechanism by
\begin{eqnarray*}
	\Rev(M) &\geq& \Pr[\mathrm{(a)}] \E_{\vec{s}}[R_1] + \Pr[\mathrm{(b)} \cap A] \E_{\vec{s},Z}[\tilde{R}_3\ |\ A]\\
	&=& \frac{1}{\alpha^2 + 4\alpha d + 1}\left((\alpha^2 + 1)\cdot \E_{\vec{s}}[R_1] + \alpha d \cdot \E_{\vec{s},Z}[\tilde{R}_3\ |\ A]\right).
\end{eqnarray*}

In addition, for any mechanism, an upper bound on the revenue is the \emph{lookahead benchmark}.  In our setting, this is the revenue from selling only to the winner at his winning conditional monopoly reserve, plus the welfare of the highest valued agent who did not win (Lemma 4.1 in \citet{CFK}). That is,
\begin{equation}
\textsc{Lookahead Benchmark:} \hspace{1cm} \opt \leq \E_{\vec{s}}[R_1 + v_2(\vec{s})]. \hspace{4cm}
\label{eq:lookaheadbenchmark}
\end{equation}
Since we will bound the value of bidder $2$ by the value of bidder $3$, we now observe that the value of bidder $3$ is bounded by the revenue from offering conditional monopoly reserves to bidders 1 or 3.
\begin{align}
\E_{\vec{s},Z}[v_3(\vec{s})\ |\ A]\ &=\ \E_{\vec{s},Z}[v_3(s_1, \vec{s}_{-1})\ - v_3(b_1^*, \vec{s}_{-1}) + v_3(b_1^*, \vec{s}_{-1}) \ |\ A] \nonumber\\
&\leq\ \E_{\vec{s},Z}\left[\begin{aligned}
&d \left( v_3(s_1, 0_{[3]}, \vec{s}_{-13}) - v_3(b_1^*, 0_{[3]}, \vec{s}_{-13}) \right) \\ &\ +\  v_3(b_1^*, \vec{s}_{-1})
\end{aligned}\ \Big\vert\ A\right]& \mbox{$d$-\concavity}\nonumber \\
&\leq\ \E_{\vec{s},Z}[d v_3(s_1, 0_{[3]}, \vec{s}_{-13}) + v_3(b_1^*, \vec{s}_{-1}) \ |\ A]& \text{Nonnegativity of $v_3$}\nonumber\\
&\leq\ \E_{\vec{s},Z}[d \tilde{R}_3  + v_3(b_1^*, \vec{s}_{-1})  \ |\ A]& \text{Definition~\ref{def:losingcondreserve} of $\tilde{R}_3$} \nonumber\\
&\leq\ d \cdot \E_{\vec{s},Z}[\tilde{R}_3  \ |\ A] +  \E_{\vec{s}}[\alpha v_1(b_1^*, \vec{s}_{-1})] & \text{$x$ $\alpha$-approximates welfare} \nonumber
\\
&\leq\ d\cdot  \E_{\vec{s},Z}[d \tilde{R}_3  \ |\ A] + \alpha \E_{\vec{s}}[R_1]. & \text{Definition~\ref{def:winningcondreserve} of $R_1$}\label{eq:last}
\end{align}
Lines 4 and 6 hold because the winning conditional monopoly reserve price for $1$ is at least $v_1(b_1^*, \vec{s}_{-1})$ and the losing conditional monopoly reserve price for $3$ is at least $v_3(s_3=0, \vec{s}_{-3})$, and in both cases, the buyer has value above this with probability 1. Line 5 holds since $1$ wins at $(b_1^*, \vec{s}_{-1})$ by definition of $b_1^*$ and an the winner in $x$ $\alpha$-approximates any other bidder.

\begin{align*}
\opt \ &\leq \ \E_{\vec{s}}[R_1 + v_2(\vec{s})] & \textsc{Lookahead Benchmark} \, \eqref{eq:lookaheadbenchmark} \\
&\leq\  \E_{\vec{s}}[R_1] +\E_{\vec{s},Z}[\alpha \, v_3(\vec{s})\ | \ A] & \text{$3 = x(\vec{s}; Z)$; $x$ $\alpha$-approximates welfare of $Z$; $2\in Z$}\\
&\leq \ (\alpha^2+1)\cdot\E_{\vec{s}}[R_1] + \alpha d\cdot \E_{\vec{s},Z}[\tilde{R}_3 \ | \ A]. & \text{$\E_{\vec{s}}[R_1]=\E_{\vec{s},Z}[R_1\ |\ A]$ and \eqref{eq:last}}
\end{align*}

All together, this gives $$\opt \leq (\alpha^2+1)\cdot\E_{\vec{s}}[R_1] + \alpha d\cdot \E_{\vec{s},Z}[\tilde{R}_3 \ | \ A] \leq \left(\alpha^2 + 4\alpha d + 1\right) \Rev(M).$$

\end{proof}

The above theorem, and Theorem~\ref{thm:2signals} from Section~\ref{sec:2signals} yield the following.
\begin{corollary}
	For $c$-single-crossing concave valuations supported on 2 signals, there exist a mechanism that gives $(c^2+4c+1)$-approximation to the optimal revenue.
\end{corollary}
\subsection{Extension to randomized mechanisms}
We first note that reduction from welfare to revenue does not automatically give a reduction from expected welfare to expected revenue, as even if some random variable $X\geq1$ has some expectation $E[X]=\alpha$, the expectation of $X^2$ might be much larger\footnote{To see this, consider a random varible $X$ that takes value 1 with probability $1-p$ and $\alpha^M$ with probability $p=\frac{\alpha-1}{\alpha^M-1}$ for some very large $M$.} than $\alpha^2$.

Instead, we state the following. Consider a randomized mechanism such that for every set $S$, the probability that $j\gets x^S(\vec{s}_S)$ has a value at $\vec{s}$ that $\alpha$-approximates the maximum valued agent at $S$ is at least $p$ (independently for each $S$). Consider the case where both $x(\vec{s})$ and $x_Z(\vec{s}_Z)$ output an allocation that is $\alpha$-approximation to the maximum valued bidder in $\{1,\ldots, n\}$ and $Z$ respectively (this happens with probability $p^2$). In this case, the analysis of the previous section follows. Therefore, conditioning on $A$ and this event gives the desired guarantee. To optimize the parameters, we run (a) with probability $\frac{\alpha^2 +1}{\alpha^2 + (4\alpha d)/p^2 + 1}$ and with probability $\frac{(4\alpha d)/p^2}{\alpha^2 + (4\alpha d)/p^2 + 1}$ run (b). Plugging in the analysis from previous section yields an $(\alpha^2 + 4\alpha d/p^2 + 1)$-approximation to the revenue.

Inspecting $\randhypercol$ in the case of concave valuations, we notice that, with probability at least $1/2$, the winner is a $2c$-approximation to the optimal welfare. To see this, consider a set of signals $\vec{s}$ where $i^*\in \argmax_{i\in\{1,\ldots,n\}} v_i(\vec{s})$. Consider an ordering $\pi$ in which some set $A$ precedes $i^*$, and some set  $B=\{1,\ldots, n\}\setminus \{A\cup i^*\}$ succeeds $i^*$ and the ordering $\pi'$ with $A$ and $B$ switched: where $B$ precedes $i^*$ and $A$ succeeds $i^*$. As stated in Section \ref{sec:random}, if the value of $i^*$ at $\vec{s^{i^*}_\pi}$ (that is, after all signals in $A\cup i^*$ are turned on) is $\alpha\cdot v_i(\vec{s})$, then concavity ensures that $i^*$'s value at $\vec{s^{i^*}_{\pi'}}$ is at least $(1-\alpha)\cdot v_i(\vec{s})$. Therefore, at least in one of the orderings $\pi$ and $\pi'$, in $i^*$'s iteration, $i^*$'s value is at least half of its value at $\vec{s}$. By the way $\randhypercol$ works, the tentative winner after this iteration must $c$-approximate $i^*$'s value, and thus $2c$-approximation to his value at $\vec{s}$.
Since permutation $\pi$ and $\pi'$ are drawn with the same probability, this gives the desired guarantee.

\begin{corollary}
	For $c$-single-crossing concave valuations supported on an arbitrary number of signals, there exist a mechanism that gives $(4c^2+32c+1)$-approximation to the optimal revenue.
\end{corollary}



\bibliographystyle{plainnat}

\bibliography{interdependent}
\newpage
\appendix

\section{Alternative Single Crossing Definitions} 
\label{sec:alternative-definitions}

We have defined our notion of $c$ single crossing as a relaxation of the single crossing definition in \citep{RTCoptimalrev}.
This definition imply additional definitions that have been used in the literature \citep{milgrom1982theory,Aspremont82,maskin1992,ausubel1999generalized,dasgupta2000efficient,
bergemann2009information,CFK,che2015efficient,li2016approximation}.
Thus, the impossibility results given in Sections \ref{sec:preliminaries} and \ref{sec:impossibility} hold for any of these definitions.

It is less clear how to define a $c$-single crossing relaxation for other variants of the single crossing definition.
The definition of single crossing in \citep{RTCoptimalrev} is closed under subsets: if a set of valuation functions have this property then so does every subset. In contrast, the  \citep{dasgupta2000efficient} definition is that
for all $i\neq j$, \begin{equation}\frac{\partial v_i}{\partial s_i}(\vec{s})\geq\frac{\partial v_j}{\partial s_i}(\vec{s})\label{eq:dasgupta}\end{equation} at any point where $v_i(\vec{s})=v_j(\vec{s})=\max_k v_k(\vec{s})$. This definition is not closed under subsets.

However, if we modify the definition so that Equation \eqref{eq:dasgupta} holds whenever $v_i(\vec{s})=v_j(\vec{s})$, the modified definition is closed under subsets. Now, one could define a new version of $c$-single crossing based upon the above: for all $i\neq j$, \begin{equation*}\frac{\partial v_i}{\partial s_i}(\vec{s})\geq\frac{1}{c}\cdot\frac{\partial v_j}{\partial s_i}(\vec{s})\end{equation*} at any point where $v_i(\vec{s})=v_j(\vec{s})$. The proofs for 2 bidders, 2 signals, and the randomized approximations, when applied to this new definition, give the same results as the ones obtained in this paper. 
\section{Exact Numbers Satisfying the Structure in Figure \ref{fig:nocapxconstraints}} \label{sec:nocapproxnumbers}
We now present exact numbers that support the structure suggested in Figure \ref{fig:nocapxconstraints} in Section \ref{sec:nocapprox}. In this example, $c = 2$. One can verify that in $\vec{s}_a=(2,0,0)$ no bidder $c$-approximates bidder~2's value, in $\vec{s}_c=(1,1,0)$ bidder~2's value does not $c$-approximate bidder~3's value, and in $\vec{s}_e=(0,1,1)$ no bidder $c$-approximates bidder~1'a value, as desired in order to yield the contradiction. 
\newline\newline
\begin{tabular}{| c | c | c |} \hline
	$(0,1,0)$ & $\vec{s}_c=(1,1,0)$ & $(2,1,0)$ \\
	$v_1 = 0.007219$ & $v_1 = 0.014529$ & $v_1 = 0.017809$ \\
	$v_2 = 0.004286$ & $v_2 = 0.008091$ & $v_2 = 0.014651$ \\
	$v_3 = 0.003180$ & $v_3 = 0.017799$ & $v_3 = 0.017809$ \\ \hline
	$(0,0,0)$ & $(1,0,0)$ & $\vec{s}_a=(2,0,0)$ \\
	$v_1 = 0$ & $v_1 = 0.000100$ & $v_1 = 0.003381$ \\
	$v_2 = 0.000676$ & $v_2 = 0.000876$ & $v_2 = 0.007436$ \\
	$v_3 = 0.003170$ & $v_3 = 0.003370$ & $v_3 = 0.003380$ \\ \hline
\end{tabular} \quad \quad \quad
\vspace{.5cm}
\begin{tabular}{| c | c | c |} \hline
	$\vec{s}_e=(0,1,1)$ & $(1,1,1)$ & $(2,1,1)$ \\
	$v_1 = 0.009449$ & $v_1 = 0.016760$ & $v_1 = 0.020040$ \\
	$v_2 = 0.004295$ & $v_2 = 0.008101$ & $v_2 = 0.014661$ \\
	$v_3 = 0.004295$ & $v_3 = 0.018915$ & $v_3 = 0.018925$ \\ \hline
	$(0,0,1)$ & $(1,0,1)$ & $(2,0,1)$ \\
	$v_1 = 0.002231$ & $v_1 = 0.002331$ & $v_1 =  0.005611$ \\
	$v_2 = 0.000686$ & $v_2 = 0.000886$ & $v_2 = 0.007446$ \\
	$v_3 = 0.004286$ & $v_3 = 0.004486$ & $v_3 = 0.004495$ \\ \hline
\end{tabular}

\section{Random Mechanisms for General $c$-SC Valuations}
\label{sec:random-appendix}

In this section we show that if the permutation $\pi$ is chosen randomly, then the approximation guarantee given for $\randhypercol$  improves significantly.
In Section \ref{sec:random_sqrt} we establish a $2c^{3/2} \sqrt{n}$-approximation guarantee, and in Section \ref{sec:randoalglb} we show that the $\sqrt{n}$ factor is inevitable (for this mechanism).

\subsection{A Random, Prior-Free, Universally Truthful $2c^{3/2} \sqrt{n}$-Approximation}
\label{sec:random_sqrt}

Given a signal profile $\vec{s}$, a permutation $\pi$, let $\vec{s^j_\pi}$ be the $j$th intermediate profile as defined in Definition \ref{def:ithintprof}.

 We define $x^j_{\pi}(\vec{s})$ to be the tentative winner for $\vec{s^j_\pi}$ set during the $j^{\rm th}$ iteration of $\hypercol(\pi)$. 

To simplify the notation hereinafter, we drop $\vec{s}$ from $x^t_{\pi}$. In Lemmas~\ref{allocatedvaluemonotonicity}, \ref{boundedalgchange}, we rename the $t^{\rm th}$ bidder in the ordering $\pi$ to $t$, and omit $\pi$ as well. Therefore, we use $\vec{s^t}$,
and $x^t$.

\begin{numberedtheorem}{\ref{thm:randrootn}}
For arbitrary $c$-single-crossing valuations,
$\randhypercol$ is a universally truthful and prior-free mechanism that gives a $2c^{3/2} \sqrt{n}$-approximation to social welfare.
\end{numberedtheorem}

The high level idea of our proof is as follows. Fix a permutation $\pi$. For any bid profile $\vec{s}$ where agent $i^*\in\argmax_i\{v_i(\vec{s})\}$ has the highest value, we show that if $\hypercol(\pi)$ performs badly, then if $i^*$ is placed sufficiently close to the end of the ordering, the allocated bidder is guaranteed to be a good approximation of $i^*$. The reason for this is the following: let $j$ be the allocated bidder under permutation $\pi$, and assume by renaming that agent $i$ is the $i^{\rm th}$ agent in the ordering. Consider the bid profile $\vec{s}=(s_1,s_2,\ldots,s_n)$, and all the intermediate bid profiles that $\hypercol(\pi)$ considers, $\vec{s^i}=(s_1,s_2,\ldots,s_i,\vec{0}_{[i+1:n]})$ for every $i$ (the bid profile where the last $i+1,\ldots,n$ signals are zeroed). If for any $i$, the value of $i^*$ increases by a lot from profile $\vec{s^{i-1}}$ to $\vec{s^{i}}$, then by $c$-single-crossing, $i$'s value must also increase.  However, $i$'s increase cannot be much larger than the value of $j$ at $\vec{s}$; otherwise, $\hypercol(\pi)$ would have allocated to $i$ at $\vec{s^{i}}$, and then $j$ would not win at $\vec{s}$. 
Therefore, there is no $i$ where $v_{i^*}(\vec{s^{i}})-v_{i^*}(\vec{s^{i-1}})$ is too large (this is cast in Lemma~\ref{boundedalgchange}).

Because the increase of $i^*$'s value cannot be too large between intermediate profiles, 
if $i^*$ was to appear toward the end of $\pi$, his value at the intermediate profile must already be a sufficiently large fraction of his value at $\vec{s}$, and thus any allocated bidder must have also been a good enough approximation to $i^*$; otherwise, $\hypercol(\pi)$ would have reallocated to $i^*$. Since $\pi$ is chosen uniformly at random, the probability that $i^*$ appears towards the end of $\pi$ is sufficiently large to guarantee a good approximation.

We use the following two lemmas to prove the improved approximation guarantee for our randomized mechanism.

The first lemma shows that the value of the final winner is (weakly) greater than the value of any tentative winner.
This is proved by showing that as the iteration number $t$ increases, the value of the tentative winner increases.

\begin{lemma} \label{allocatedvaluemonotonicity}
Given a profile $\vec{s}$ and permutation $\pi$, and assume $\hypercol(\pi)$ allocates to $j$ at $\vec{s}$.
For any iteration $t$, $v_{x^t}(\vec{s^t})\leq v_{j}(\vec{s})$.
\end{lemma}

\begin{proof}
	Observe that $\hypercol(\pi)$ has the following properties: (a) whenever the algorithm reallocates at some profile, the value of the new tentative winner is larger than the previous one, and (b) if the algorithm does not reallocate as we increase $s_t$ during iteration $t$, then the tentative winner's value (weakly) increases.
	
	We prove the lemma by backwards induction on $t$. When $t=n$, $x_t=j$ and $\vec{s^t}=\vec{s}$, so the claim is a tautology. Suppose that for $t\leq n$, $v_{x^t}(\vec{s^t})\leq v_{j}(\vec{s})$. To finish the proof, one needs to show that
	\begin{eqnarray}
		v_{x^{t-1}}(\vec{s^{t-1}})\leq v_{x^t}(\vec{s^t}).\label{eq:induction_step}
	\end{eqnarray}
	At iteration $t$, one of the following must happen:
	\begin{itemize}
		\item Either the tentative winner stayed the same, i.e., $x_{t-1}=x_t$, and \eqref{eq:induction_step} holds by (b), or
		\item We reallocate to $t$ at some $\vec{s'}=(s_1,\ldots, s_{t-1}, s'_{t}, 0_{[t+1:n]})$ for $s'_{t}\leq s_{t}$. By (b), we have that $v_{x_{t-1}}(\vec{s^{t-1}})\leq v_{x_{t-1}}(\vec{s'})$, by (a) $v_{x_{t-1}}(\vec{s'})\leq v_{x_{t}}(\vec{s'})$, and by (b) $v_{x_{t}}(\vec{s'})\leq v_{x_{t}}(\vec{s^t})$. Chaining the last three inequalities yields \eqref{eq:induction_step}.
	\end{itemize}
	Hence, \eqref{eq:induction_step} holds, and the lemma follows.
\end{proof}

The second lemma formalizes the intuition that at every iteration of $\hypercol$, $i^*$'s increase in value cannot be large compared to the winners' value.

\begin{lemma} [Change Bounded by The Allocated Value] \label{boundedalgchange}
	Suppose the final allocation at $\vec{s}$ is to agent $j$ and the highest value at $\vec{s}$ is of agent $i^* \in \argmax _i v_i(\vec{s})$. Then for every iteration $t$, $$v_{i^*}(\vec{s^t})-v_{i^*}(\vec{s^{t-1}})\leq c^2 v_j (\vec{s}).$$
\end{lemma}

\begin{proof}
First, notice that by the definition of $c$-single-crossing,
\begin{eqnarray}
	v_{i^*}(\vec{s^t})-v_{i^*}(\vec{s^{t-1}})\leq c\left(v_{t}(\vec{s^t})-v_{t}(\vec{s^{t-1}})\right).\label{eq:a}
\end{eqnarray}

To finish the proof, we argue that
\begin{eqnarray}
	v_{t}(\vec{s^t})\leq c v_j (\vec{s}).\label{eq:b}
\end{eqnarray}
By the non-negativity of $v_t$, combining \eqref{eq:a} with \eqref{eq:b} is enough to to prove the lemma. Recall that by renaming, agent $i$ is at position $i$ in the ordering $\pi$. To show \eqref{eq:b}, we consider two cases.

\textit{Case 1:} If $t > j$, then to wind up with $x(\vec{s}) = j$, after iteration $t$, we must have had $x^t=j$. Thus, it must be the case that $v_{t}(\vec{s^t})\leq cv_j (\vec{s^t})$, or $\hypercol$ would have reallocated to bidder $t$ during iteration $t$. By monotonicity, $v_j (\vec{s^t})\leq v_j(\vec{s})$,  so \eqref{eq:b} follows.

\textit{Case 2}: If $t \leq j$, then if $j' = x^{t}$ is the tentative winner after iteration  $t$, it must be the case that either $j'=t$, or $v_{t}(\vec{s^{t}})\leq c \cdot v_{j'}(\vec{s^{t}})$. Otherwise, $\hypercol$ would have reallocated to bidder $t$ at that iteration. Since by Lemma~\ref{allocatedvaluemonotonicity},  $v_{j'}(\vec{s^{t}})\leq v_j(\vec{s})$, \eqref{eq:b} follows in this case as well.

We conclude that in both cases,  \eqref{eq:b} follows, and hence, so does the lemma.
\end{proof}

We now proceed to prove the theorem.

\begin{proof}[Proof of Theorem~\ref{thm:randrootn}]
(Monotonicity.) $\randhypercol$ is clearly universally truthful: the only random choices of the algorithm are in selecting an ordering $\pi$ uniformly at random, which is done independent of the bid profile $\vec{s}$. Given $\pi$, the algorithm returns the allocation rule $x$ from  $\hypercol(\pi)$, which we have already proven is monotone.

(Approximation.) For any given signal profile $\vec{s}$, let $i^* \in \argmax _i v_i(\vec{s})$ be the bidder with the highest value at $\vec{s}$.
Observe that we can draw the permutation $\pi$ used in $\randhypercol$ in the following manner.  First, we draw the relative positioning of every bidder except $i^*$ (using a uniformly at random permutation), denoted by $\pi_{-i^*}$. Rename the bidders such that $\pi_{-i^*} = (1, \ldots, n-1)$.  Then, draw the position of $i^*$ by drawing a uniform number in $t \in \{1,\ldots,n\}$, giving $\pi = (1, \ldots, t-1, i^*, t, \ldots, n-1)$.  This permutation is drawn uniformly at random, thus satisfying the requirements of $\randhypercol$.

We show that for every $\pi_{-i^*}$, $i^*$'s value at $\vec{s}$ is well-approximated a large fraction of the iterations.  For the remainder of the proof, let $\pi_{-i^*}$ be fixed, so the choice of $t$ defines $\pi$.

For all such orderings with $\pi_{-i^*}$ fixed, let $\pi'$ be the ordering where $i^*$ is in the worst position with respect to how $i^*$'s value at $\vec{s}$ is approximated.  That is, let $t'$ be the position of $i^*$, where $\pi' = (1, \ldots, t'-1, i^*, t', \ldots, n-1)$.  Let $j'$ be the bidder who is allocated at $\vec{s}$ in $\hypercol(\pi')$; $j'$'s value is minimum over all potential winners $j$ in $\hypercol(\pi)$ with $\pi_{-i^*}$ fixed.  Let $\alpha c n$ be the approximation of $j'$'s value to $i^*$ where $\alpha <1$, that is,
\begin{eqnarray}
v_{i^*}(\vec{s}) = \alpha cn \cdot v_{j'}(\vec{s}).\label{eq:acn_approx}
\end{eqnarray}

Consider the iteration $z=n-\frac{\alpha n}{2c}$ under $\pi'$. We have that
\begin{align}
	v_{i^*}(\vec{s^z_{\pi'}}) &= v_{i^*}(\vec{s})-\sum_{\ell=z + 1}^{n}\left(v_{i^*}(\vec{s^{\ell}_{\pi'}})-v_{i^*}(\vec{s^{\ell-1}_{\pi'}})\right) &  \nonumber\\
	&\geq v_{i^*}(\vec{s})- \sum_{\ell=n - \frac{\alpha n}{2c} + 1}^{n} c^2 v_{j'}(\vec{s}) & \text{by Lemma~\ref{boundedalgchange}}\nonumber\\
	&= v_{i^*}(\vec{s})- \frac{1}{2}\alpha cn \cdot v_{j'}(\vec{s}) &  \nonumber\\
	&= \frac{v_{i^*}(\vec{s})}{2}.& \text{by~\eqref{eq:acn_approx}} \label{eqn:istarlarge}
\end{align}

Note that, coordinate-wise, it holds that
$$\vec{s^z_{\pi'}} \leq (\vec{s}_{[1:z]}, \vec{0}_{[z+1:n-1]}, s_{i^*}).$$
To see this, note that if $t'>z$, then $\vec{s^z_{\pi'}} = (\vec{s}_{[1:z]}, \vec{0}_{[z+1:n-1]}, 0)$, and if $t' \leq z$, then
$\vec{s^z_{\pi'}} =  (\vec{s}_{[1:z-1]}, \vec{0}_{[z:n-1]}, s_{i^*})$.

Let
$$\vec{s^{t}_{\pi}} = (\vec{s}_{[1:t-1]}, \vec{0}_{[t:n-1]}, s_{i^*})$$
be the ordering in which $i^*$ is in position $t$ (fixing $\pi_{-i^*}$).

By definition of $\hypercol(\pi)$, if the tentative winner at profile $\vec{s_\pi^{t}}$ after iteration $t$ is $\tilde{j}$, then
\begin{eqnarray}
	v_{\tilde{j}}(\vec{s_\pi^{t}})\geq v_{i^*}(\vec{s_\pi^{t}})/c,\label{eq:jp_approx}
\end{eqnarray}
as otherwise the algorithm would have reallocated to $i^*$.

In addition, in the event where $i^*$ is among the last $\frac{\alpha n}{2c}$ bidders (i.e., $t > z$), by monotonicity of $v_{i^*}$ and by~\eqref{eqn:istarlarge},
\begin{eqnarray}
	v_{i^*}(\vec{s^{t}_{\pi}})\geq v_{i^*}(\vec{s^z_{\pi'}})\geq  \frac{v_{i^*}(\vec{s})}{2}.\label{eq:largeatpi}
\end{eqnarray}

Let $j$ be the allocated bidder at $\vec{s}$ in $\hypercol(\pi)$, then

\begin{align}
	\E_{t}[v_{j}(\vec{s})] &\geq \E_{t}[v_{\tilde{j}}(\vec{s^{t}_\pi})]  & \text{by Lemma~\ref{allocatedvaluemonotonicity}}  \nonumber\\
	&\geq \E_{t}[v_{i^*}(\vec{s^{t}_\pi})/c]  & \text{by \eqref{eq:jp_approx}}\nonumber\\
	 &\geq \frac{1}{c}\E_{t}[v_{i^*}(\vec{s^{t}_\pi})|t>z]\cdot\Pr_{t}[t>z]  &  \nonumber\\
	 &\geq \frac{1}{c}\cdot \frac{v_{i^*}(\vec{s})}{2}\cdot\frac{\alpha n }{2c n}  & \text{by \eqref{eq:largeatpi} and the random choice of $t$} \nonumber\\
	&= \frac{v_{i^*}(\vec{s})}{\frac{4c^2}{\alpha}}.& \label{eq:expected_welfare}
\end{align}
Therefore, the approximation ratio is at most $\frac{4c^2}{\alpha}$. However, since at the worst position of $i^*$, the approximation ratio is $\alpha c n$, then at a random position, the approximation ratio is at most $\alpha cn$. Therefore, in expectation, we get a $\min\{4c^2/\alpha, \alpha cn\}$-approximation to $i^*$'s value at $\vec{s}$.  This is maximized, i.e., gets the worst possible ratio, when $\alpha = \sqrt{4c/n}$, giving a $2c^{3/2} \sqrt{n}$-approximation.

Since this holds for \emph{every} possible $\pi_{-i^*}$ that is drawn, then overall we have a $2c^{3/2} \sqrt{n}$-approximation.
\end{proof}


\subsection{Lower Bound of \lowercase{$\Omega({c\sqrt{n}/\log{n}})$} for $\randhypercol$}\label{sec:randoalglb}

We now show a lower bound of $\Omega({c\sqrt{n}/\log{n}})$ for $\randhypercol$.

Consider the following example.  There are $n+1$ agents---$n$ of them are  partitioned into $\frac{\sqrt{n}}{\log{n}}$ fixed sets $S$ of size $\log{n} \sqrt{n}$, plus some lonesome $i^*$.  Each bidder has a low (0) and high (1) signal.

The valuation functions are as follows.  For agent $i \in S$, $v_i(\vec{s}) = \mathbbm{1}[s_j = 1 \forall j \in S]$, that is, his value is 1 if and only if every agent in his set has a high signal, and is otherwise 0.  Agent $i^*$ has a value that is $c$ times the number of sets where all agents are high, that is, $v_i^*(\vec{s}) = c \cdot |\{S : s_j = 1 \forall j \in S\}|$.

Clearly it also satisfies $c$-single-crossing: as any bidder $i \neq i^*$ changes his signal from low to high, if his value changes, it does by 1.  So do the values of everyone else in his set, so they're clearly single-crossing.  So does $i^*$'s value by $c$, which is clearly $c$-single-crossing.  If his value does not change, then he does not make an entire set high, nor does he increase the number of sets that are high, so no one's value changes, which is also $c$-single-crossing.  None of the agents have signals that depend on $i^*$, so his value is trivially $c$-single-crossing.

For any set $S$, the probability over the random ordering that none of its agents appear in the last $\sqrt{n}$ in the permutation is $$\frac{\binom{n-\sqrt{n}}{\log n\sqrt{n}}}{\binom{n}{\log n\sqrt{n}}}\leq \left(\frac{n-\sqrt{n}}{n}\right)^{\log{n} \sqrt{n}} = \left(1 - \frac{1}{\sqrt{n}} \right)^{\log{n} \sqrt{n}} \leq e^{-\log{n}} = \frac{1}{n}.$$

Then, over all $\frac{\sqrt{n}}{\log{n}}$ sets, the probability that any of the sets have no agents appear in the end is, by a union bound, at most $\frac{\sqrt{n}}{\log{n}}  \cdot \frac{1}{n}$.  That is, with probability at least $1 - \frac{1}{\sqrt{n} \log n}$, none of the sets have no agents appear at the end, or equivalently, every set has some agent who comes in the last $\sqrt{n}$ of the permutation, and thus at any cell allocated in the first $n - \sqrt{n}$ iterations, $i^*$ has a valuation of 0.

Then if this event occurs and $i^*$ comes in the first $n- \sqrt{n}$ in the permutation, $i^*$ has a valuation of $0$, so we will not allocate to $i^*$, and thus at the end, in cell $\vec{1}$, we must have allocated to a bidder $i \neq i^*$ who has value 1, making our approximation off by a factor of $c |S| = c \frac{\sqrt{n}}{\log{n}}$.

That is, with probability at least $\frac{n-\sqrt{n}}{n} \cdot \left( 1 - \frac{1}{\sqrt{n} \log n} \right)$ our approximation is off by $c \frac{\sqrt{n}}{\log{n}}$.  We can lower bound this probability by $1 - \frac{2}{\sqrt{n}}$.

Then at best, our approximation is $$\left(1 - \frac{2}{\sqrt{n}} \right) \frac{\log{n}}{c \sqrt{n}} + \frac{2}{\sqrt{n}} = \frac{(\sqrt{n} - 2) \log{n} + \sqrt{n}2c}{cn} \leq \frac{\sqrt{n} \log{n} + \sqrt{n}2c}{cn}=\frac{\log{n} + 2c}{c\sqrt{n}}.$$
For $c< \log{n}$, this gives a lower bound of at least $\sqrt{n}c/\log{n}$.

\end{document}